\documentclass{article}
\usepackage[margin=1.25in]{geometry}
\usepackage{amsmath, amssymb, amsthm, mathtools,amscd,mathalpha,amsfonts,mathrsfs,tikz-cd}
\usepackage{cite}
\usepackage{listings}
\usepackage{enumerate}
\usepackage[parfill]{parskip}
\usepackage{hyperref}
\usepackage{xcolor}
\usepackage{appendix}
\usepackage{graphicx}
\usepackage{caption}
\usepackage{subcaption}
\usepackage[font={it}]{caption}
\usepackage{float}

\DeclarePairedDelimiter\angles\langle\rangle

\DeclareMathOperator\R{\mathbb{R}}
\DeclareMathOperator\C{\mathbb{C}}
\DeclareMathOperator\Z{\mathbb{Z}}
\DeclareMathOperator\N{\mathbb{N}}

\DeclareMathOperator\Hom{Hom}

\DeclareMathOperator\supp{supp}

\DeclareMathOperator\tr{tr}

\DeclareMathOperator\Ad{Ad}
\DeclareMathOperator\ad{ad}
\DeclareMathOperator\diam{diam}

\DeclareMathOperator{\res}{res}
\DeclareMathOperator{\Tot}{Tot}

\DeclareMathOperator{\sgn}{sgn}

\makeatletter

\newcommand{\subsetsim}{\mathrel{\mathpalette\subset@sim\relax}}

\newcommand{\subset@sim}[2]{%
  \vtop{\offinterlineskip\m@th
    \ialign{\hfil$#1##$\hfil\cr
      \subset\cr\noalign{\kern1pt}\sim\cr
    }%
  }%
}
\makeatother

\newcommand{\CB}{{\mathcal B}}

\newcommand{\CO}{{\mathcal O}}

\newcommand{\ZZ}{{\mathbb Z}}
\newcommand{\RR}{{\mathbb R}}

\newcommand{\SA}{{\mathscr A}}

\newcommand{\SAal}{{\mathscr A}_{a \ell}}

\newcommand{\mfkDal}{{\mathfrak{D}}_{al}}

\newcommand{\mfkDpal}{{\mathfrak D}^{\psi}_{al}}

\newcommand{\chf}{{\mathsf f}}

\newcommand{\derA}{{\mathsf A}}
\newcommand{\derB}{{\mathsf B}}
\newcommand{\derC}{{\mathsf C}}

\newcommand{\derE}{{\mathsf E}}
\newcommand{\derF}{{\mathsf F}}
\newcommand{\derG}{{\mathsf G}}
\newcommand{\derH}{{\mathsf H}}
\newcommand{\derQ}{{\mathsf Q}}

\newcommand{\derf}{{\mathsf f}}
\newcommand{\derg}{{\mathsf g}}
\newcommand{\derh}{{\mathsf h}}
\newcommand{\dert}{{\mathsf t}}
\newcommand{\derk}{{\mathsf k}}
\newcommand{\derq}{{\mathsf q}}

\newcommand{\bUM}{{\underline{U(1)}_M}}

\newcommand{\ra}{\rightarrow}

\newtheorem{theorem}{Theorem}[section]

\newtheorem{prop}{Proposition}[section]
\newtheorem{lemma}{Lemma}[section]

\newtheorem*{lemma*}{Lemma}
\newtheorem{remark}{Remark}
\newtheorem{definition}{Definition}

\newtheoremstyle{example}{}{}{}{}{\bfseries}{\smallskip}{\newline}{}

\begin{document}
\title{Quantization of the higher Berry curvature and the higher Thouless pump}
\author{Adam Artymowicz, Anton Kapustin, Nikita Sopenko
\smallskip\\
{\it California Institute of Technology, Pasadena, CA 91125, USA}}
\maketitle
\begin{abstract}
    We show that for families of 1d lattice systems in an invertible phase, the cohomology class of the higher Berry curvature can be refined to an integral degree-3 class on the parameter space. Similarly, for families of $U(1)$-invariant 2d lattice systems in an invertible phase, the higher Thouless pump can be refined to an integral degree-2 class on the parameter space. We show that the 2d Thouless pump can be identified with an excess Berry curvature of a flux insertion.
\end{abstract}
\tableofcontents

\section{Introduction}

A smooth family of gapped Hamiltonians on a finite-dimensional Hilbert space defines a smooth vector bundle on the parameter space: the bundle of ground states \cite{BarrySimon}. Its Chern classes are topological invariants of the family taking values in the integral cohomology of the parameter space. Their de Rham representatives are closed differential forms which are  polynomials in the curvature of the celebrated Berry connection \cite{MBerry}. The simplest of them is the trace of the Berry curvature divided by $2\pi$, which is a closed 2-form with integral periods. It is the de Rham representative of the 1st Chern class of the bundle of ground states. 

It is of considerable interest to generalize the Berry connection and the associated topological invariants to families of gapped infinite-volume systems in $d$ spatial dimensions. From the field theory viewpoint, such topological invariants should describe  topological terms in the effective action for $(d+1)$-dimensional $\sigma$-model obtained by integrating out the gapped degrees of freedom. If a continuous symmetry $G$ is present for all values of the parameters, these terms may also depend on the gauge field for $G$. Such topological terms are known as (equivariant) Wess-Zumino-Witten terms (see \cite{WZWreview} for a brief review), and the Berry connection can be viewed as a special case corresponding to $d=0$ and trivial symmetry group. Non-trivial Wess-Zumino-Witten terms signal the presence of gapless loci in the parameter space \cite{diabolical}. They also probe the topology of the space of gapped systems and thus can be used to test the Kitaev conjecture which posits that spaces of ``invertible" gapped systems in all dimensions fit into a loop spectrum in the sense of homotopy theory \cite{Kitaevlecture}.

Recent works \cite{KapSpoBerry,KapSpoThouless,LocalNoether} constructed some topological invariants of smooth families of infinite-volume gapped lattice systems  They showed how to assign a de Rham class $[\omega^{(d+2)}] \in H^{d+2}(M,\RR)$ to a smooth family of lattice systems on $\RR^d$ parameterized by $M$. Since for $d=0$ this class  reduces to the cohomology class of the curvature of the Berry connection, the generalization to $d>0$ is called the higher Berry class. By the usual Chern-Weil theory, the cohomology class of the Berry curvature is an obstruction for the existence of a global trivialization of the bundle of ground states. Similarly, the higher Berry class is an obstruction for the existence of a smooth family of automorphisms which maps the family of ground states to a constant family \cite{LocalNoether}.
For $G$-equivariant families, where $G$ is a compact connected Lie group, there is an equivariant refinement of higher Berry classes taking values in the equivariant cohomology $H^{d+2}_G(M,\RR)$ \cite{LocalNoether}. The higher Berry classes, as well as the equivariant higher Berry classes for $G = U(1)$, are reviewed in Section \ref{sec:berrydef} below. 

By analogy with the $d=0$ case, one may ask if higher Berry classes are ``quantized'', or more precisely, if they can be refined to integral cohomology classes. A simple argument shows that this is not possible for arbitrary families of gapped systems. Let $d=2$, $M=\{pt\}$ and $G=U(1)$. In this case the higher Berry class takes values in $H^4_{U(1)}({pt},\RR)\simeq \RR$ and is proportional to the Hall conductance \cite{LocalNoether}. It is well known that the Hall conductance of 2d gapped systems is not quantized, in general \cite{LaughlinFQHE}. Nevertheless, it can be shown to be quantized for short-range entangled systems, or more generally, for systems in an invertible phase \cite{hastingsmichalakis,bachmannetal,bachmann2019many,ThoulessHall}. One might hope that for such systems all higher Berry classes can be refined to integral cohomology classes. The only other case where this was shown to be true is $d=1$, $G$ a compact topological group, and $M=S^1$ (with $G$ acting trivially on $S^1$), where the equivariant higher Berry class measures the net charge pumped across a section of a 1d system under a periodic variation of parameters \cite{bachmannetal, BachmannEtAl1DThouless, ThoulessHall}. This quantity is known as the Thouless pump \cite{Thouless}.\footnote{Quantization of the Thouless pump holds for arbitrary gapped 1d families. This is not surprising, since all gapped 1d systems are believed to be short-range entangled.}

In this paper, we show how to construct integral refinements of higher Berry classes in two other interesting situations. The first one (Theorem \ref{thm:1dBerryquantization}) is $d=1$, $G$ trivial, $M$ arbitrary. In this case the higher Berry class takes values in $H^3(M,\RR)$ and we show how to refine it to a class in $H^3(M,\ZZ)$ for families of invertible 1d systems. At least for $M=S^3$, the integrality of the higher Berry class is very natural since it measures the flow of ordinary Berry curvature in a cyclic process \cite{HermeleEtAl}. The second case (Theorem \ref{thm:2dThoulessquantization}) is $d=2$, $G=U(1)$, $M$ arbitrary with a trivial $U(1)$ action. In this case the equivariant higher Berry class takes values in $H^4_{U(1)}(M,\RR)=H^4(M,\RR)\oplus H^2(M,\RR)\oplus H^0(M,\RR)$, where the three components correspond to the non-equivariant higher Berry class, the 2d generalization of the Thouless pump, and the Hall conductance, respectively. We show that for invertible 2d systems the $H^2(M,\RR)$ component can be refined to a class in $H^2(M,\ZZ)$.

Our proof of Theorem \ref{thm:2dThoulessquantization} is based on a new physical interpretation of the 2d Thouless pump as the Berry curvature of a fluxon\footnote{This should be compared to Laughlin's interpretation of the Hall conductance as the charge of a fluxon -- see \cite{Laughlin} for the original argument and \cite{ThoulessHall} for a version Laughlin's argument in the formalism used in this work.}. Given any $U(1)$-invariant state, one can obtain a new $U(1)$-invariant state by inserting a $2\pi$ flux. We will always choose the gauge transformation producing the flux insertion to be concentrated on a line in physical space terminating at the flux insertion point, which we will call the Dirac string. Given a family $\psi_M$ of gapped $U(1)$-invariant systems (with a fixed $U(1)$ action) parameterized by $M$, one may form a new family $\psi_{fluxon}$ over $M$ by performing a flux insertion on each state in the family $\psi_M$. Since $\psi_M$ and $\psi_{fluxon}$ are families of 2-dimensional states, their (ordinary) Berry curvatures are divergent, but because the flux insertion is a point-like object, the \textit{excess} Berry curvature of $\psi_{fluxon}$ should be a well-defined 2-form on $M$. We obtain an expression for it as follows. Performing the flux insertion continuously, we have a family $\psi$ of states on $M\times I$, $I=[0,2\pi]$, which restricts to $\psi_M$ and $\psi_{fluxon}$ on $\partial(M\times I) = M\sqcup M$ (see Figure \ref{fig:psigamma} in Section \ref{sec:2dThoulessQuant}). Let $D$ be a large disc in physical space containing the point where the flux insertion occurs. If $\nu \in \Omega^{3}(M\times I)$ is a 3-form which measures the current of ordinary Berry curvature flowing into the disc $D$ then the excess Berry curvature of $\psi_{fluxon}$ is given by the fiber integral $\int_I\nu\in H^2(M,\RR)$. Since we are taking a fiber integral of $\nu$, we are only interested in its vertical component\footnote{The space of vertical forms on $M\times I$ is the quotient of $\Omega^\bullet(M\times I)$ by those $\omega$ for which $\iota_{X}\omega=0$ for $X=\frac{\partial}{\partial\theta}$.}, which we will call $\nu_{vert}$. This component contains contributions only from the point at which the Dirac string intersects the boundary of $D$, and is thus $O(1)$ in the size of $D$. By contrast, the other components of $\nu$ will contain contributions from the whole boundary of $D$.

In this work, we perform the flux insertion ``at infinity''. This involves moving the flux insertion point off to infinity so that the Dirac string goes along the $y$-axis without terminating (see Figure \ref{fig:psirho} in Section \ref{sec:2dThoulessQuant}). $\psi$ is now a family of states on $M\times S^1$ such that for $x\in M$ and $\theta\in S^1$, $\psi_{(x,\theta)}$ equals $(\psi_M)_{x}$ with a $\theta$-domain wall inserted on the $y$-axis, and $\nu$ measures the Berry curvature pumped along the domain wall. As discussed in \cite{HermeleEtAl}, the pumping of Berry curvature is given by the higher Berry form $\omega^{(3)}$, and so $\nu = \omega^{(3)}$. Only the vertical component $\nu_{vert}$ of this form is well-defined, and in Section \ref{sec:thouless-as-berry} we extract this component and show that it equals $\mu\wedge d\theta$, where $\mu\in \Omega^2(M)$ is a representative of the 2d Thouless pump invariant. In Section \ref{sec:2dThoulessQuant} we show that if the states in question are invertible then the proof of ordinary 1d Berry curvature quantization can be adapted to show that $\nu_{vert}$, and thus the 2d Thouless invariant, is quantized.

Note that the relation $\nu_{vert} = \mu\wedge d\theta$ holds regardless of whether the family of systems in question is in an invertible phase or not. In particular, for topologically ordered 2d systems it may happen that the excess Berry curvature associated with a flux insertion has periods which are fractions of $2\pi$. This is similar to how the charge of a flux insertion can be a fraction if the system is topologically ordered \cite{LaughlinFQHE}.

A construction of the integral refinement of the higher Berry class for 1d systems was announced in \cite{LocalNoether}. A similar integral invariant was also introduced in the thesis \cite{SpeigelThesis}, but the connection to the higher Berry curvature was not proven there. We were informed by Y. Ogata about a different approach to the integral higher Berry class for continuous families of 1d and 2d spin systems \cite{OgataKubota}. While this paper was in preparation, there appeared two papers which discuss the integral higher Berry class in the context of Matrix Product States \cite{MPS1,MPS2}. 

\noindent
{\bf Acknowledgements:}
This research was supported in part by the U.S.\ Department of Energy, Office of Science, Office of High Energy Physics, under Award Number DE-SC0011632. A.K. was also supported by the Simons Investigator Award.\\
{\bf Declarations:}
\textit{Data availability statement}. Data sharing not applicable to this article as no datasets were generated or analysed during the current study. \\
\textit{Conflict of interest}. The authors declare that they have no conflicts of interest.
\\

\section{Framework}
A recent work \cite{LocalNoether} by the last two authors introduced a framework for studying smoothly varying families of ground states of infinite-volume quantum statistical mechanical systems, and showed how this framework produces certain cohomological invariants, the (equivariant) higher Berry invariants. Since these are the subject of this paper, we begin by recalling this machinery. The reader is referred to \cite{LocalNoether} for all proofs.

\subsection{Observables and derivations}
We will be working with quantum lattice systems on the lattice $\Z^d$ for $d>0$ (we will only need $d=1,2$ in what follows)\footnote{All  results can be easily generalized to lattice systems whose sites are an arbitrary Delone subset of $\RR^d$.}. We endow the lattice $\Z^d$ with the $L^\infty$ metric\footnote{This is done purely for convenience. In \cite{LocalNoether} Euclidean metric is used, but since each of the two metrics on $\RR^d$ is upper-bounded by a multiple of the other, all the results in \cite{LocalNoether} remain true for the $L^\infty$ metric.} which we denote $d(\cdot, \cdot)$, and for a site $j$ and an integer $r$ we denote by $B_j(r)$ the ball of radius $r$ around $j$. For a subset $\Lambda\subset \Z^d$ we define $\Lambda(r):=\{j\in \Z^d: d(j,\Lambda) \le r\}$.
	
Fix $D>0$ and associate to each site $j\in \Z^d$ the C$^*$-algebra $\mathscr{A}_j := \mathcal{B}(\C^D)$ of linear operators on $\C^D$. To any finite subset $X$ of $\Z^d$ we associate the algebra $\mathscr{A}_X := \bigotimes_{j\in X}\mathscr{A}_j$. For $X\subset X'$ we have the norm-preserving inclusion $\mathscr{A}_X\subset \mathscr{A}_{X'}$ taking $A\mapsto A\otimes \boldsymbol{1}$ and the union of the resulting net $\mathscr{A}_\ell= \cup_{X}\mathscr{A}_X$ is the \textbf{algebra of local observables}. It is a normed $*$-algebra which can be completed to form the $C^*$-algebra $\mathscr{A}:=\overline{\mathscr{A}_\ell}^{\|\cdot \|}$ which we call the \textbf{algebra of quasilocal observables}. For an infinite subset $X\subset \Z^d$ we write $\mathscr{A}_{X}$ for the norm closure of $\cup\mathscr{A}_{Y\subset X}$, where $Y$ ranges over all finite subsets of $X$. 
	
We denote by $\overline{\tr}(A)$ the state which is defined on local observables $A \in\mathscr{A}_X$ by $\frac{1}{D^{|X|}}\overline{\tr}(A)$, where $|X|$ denotes the cardinality of $X$. We also have, for any (possibly infinite) subset $X\subset \Z^d$ the conditional expectation $\overline{\tr}_{X}$ onto $\mathscr{A}_{X^c}$, which is defined by the condition $\overline{\tr}_{X}(A\otimes B) := \overline{\tr}(A)B$ whenever $A\in \mathscr{A}_Y$ and $B\in \mathscr{A}_Z$ with $Y\subset X$ and $Z\subset X^c$. 

Given two algebras of quasilocal observables $\SA$ and $\SA'$ with on-site dimensions $D$ and $D'$ we define $\SA\otimes\SA'$ to be the norm closure of $\otimes_j (\SA_j\otimes\SA'_j)$. It is also a quasilocal algebra, with on-site dimension $DD'$. Physically, it corresponds to the ``stacking'' of two lattice systems described by algebras $\SA$ and $\SA'$.

In what follows we will rarely refer to the quasilocal algebra. Instead, we will mostly work with a subalgebra of $\mathscr{A}$ obtained by imposing a stricter notion of locality. For each $j\in \Z^d$ and $\alpha \in \Z_{\ge 0}$ we may define a seminorm on $\mathscr{A}_{\ell}$ by 
\begin{align}
    \|A\|_{j,\alpha} := \|A\| + \sup_r(1+r)^\alpha \inf_{B\in \mathscr{A}_{B_j(r)}}\|A-B\|.
\end{align}
Fixing any $j\in \Z^d$ and allowing $\alpha\in \Z_{\ge 0}$ to vary we obtain a family of seminorms. The completion $\mathscr{A}_{al}:= \overline{\mathscr{A}}^{\|\cdot\|_{j,\cdot}}$ with respect to this family of seminorms, which we term the \textbf{algebra of almost-local observables}, is a Fr\'{e}chet space. The seminorms $\|\cdot \|_{j,\alpha}$ and $\|\cdot \|_{k,\alpha}$ are equivalent for any $j,k\in \Z^d$ so the resulting space and its topology do not depend on the choice of $j$.
	
For any $f:\R_{\ge 0}\to \R_{\ge 0}$ we say a quasilocal observable $A \in \mathscr{A}$ is $f$-confined  at a site $j\in \Z^d$ if $\inf_{B\in \mathscr{A}_{B_j(r)}}\|A-B\|\le f(r)$ for all $r\in \Z_{\ge 0}$. A function $f:\R_{\ge 0}\to \R_{\ge 0}$ has \textbf{superpolynomial decay} if $f(r)r^\alpha\to 0$ for all $\alpha\in \Z_{\ge 0}$, and $\mathscr{A}_{al}$ can alternatively be characterized as the set of quasilocal observables that are $f$-confined  on a site $j$ for some function $f$ with superpolynomial decay.

In the most common approach to lattice systems in infinite volume, time-evolution is implemented by a strongly-continuous one-parameter family of automorphisms of $\mathscr{A}$. The generator of these automorphisms is a densely-defined derivation of $\mathscr{A}$. The appearance of densely-defined derivations is unavoidable because $\mathscr{A}$ has no nonzero globally defined outer derivations \cite{Sakai}, and the generator of time-evolution is typically outer. From this perspective the subalgebra $\mathscr{A}_{al}$ acts as a minimal domain of definition for physically relevant derivations\footnote{An analogy can be made with $C^\infty(\R)$, which embeds into larger function spaces like $C(\R)$ (the C$^*$ algebra of continuous functions on $\R$) as the minimal domain of definition for all differential operators.}. In contrast to $\mathscr{A}$, the algebra $\mathscr{A}_{al}$ has many interesting globally-defined outer derivations, some of which we describe here. Below, all derivations are taken to satisfy $D(A^*) = D(A)^*$.

We call a \textbf{brick} in $\Z^d$ any subset of the form $X = \Z^d\cap \prod_{k=1}^d[a_i,b_i]$ for some $a_i,b_i\in \Z$ with $-\infty < a_i < b_i < \infty$. Let $\derF$ be a derivation of $\mathscr{A}_{al}$. Given a brick $X$, $\overline{\tr}_{X^c}\circ \derF|_{\mathscr{A}_X}$ is a derivation of $\mathscr{A}_X$ which is necessarily equal to conjugation by a unique traceless skew-adjoint element of $\mathscr{A}_X$ which we call $\derF_X$. Each $\mathscr{A}_X$ has an inner product given by $(A,B):=\overline{\tr}(A^*B)$, and we define $\derF^X$ as the projection of $\derF_X$ onto the orthogonal complement of $\bigoplus_{Y}\mathscr{A}_Y$, where $Y$ ranges over all bricks strictly contained in $X$. This way we have $\derF_X = \sum_{Y\subset X} \derF^Y$, and the \textbf{brick decomposition} of $\derF$ is the formal sum
\begin{align}
    \derF = \sum_X\derF^X \label{eqn:brick-decomp}
\end{align}
with $X$ ranging over all bricks in $\Z^d$. Using brick decompositions one can define a family of seminorms, indexed by $\alpha\in\Z_{\ge 0}$, on the space of derivations:
\begin{align}
    \|\derF\|_{\alpha} := \sup_{X} (1+\diam(X))^\alpha\|\derF^X\|, \label{eqn:ual-seminorm}
\end{align}
and we call a derivation \textbf{uniformly almost-local} (UAL) if $\|\derF\|_\alpha<\infty$ for all $\alpha\in\Z_{\ge 0}$. We denote the space of UAL derivations by $\mfkDal$ -- it is a Fr\'{e}chet space with respect to the locally convex topology generated by these seminorms.  Furthermore, for any $\derF\in \mfkDal$ and any $A\in \mathscr{A}_{al}$, the sum $\derF(A) = \sum_X\derF^X(A)$ is absolutely convergent in the Fr\'{e}chet topology on $\mathscr{A}_{al}$, and in particular in the uniform topology. It follows that $\overline{\tr}(\derF(A))=0$ for any $\derF\in \mfkDal$ and any $A\in \mathscr{A}_{al}$.

The space of UAL derivations admits a certain kind of resolution by (higher) currents which we now describe. For $n>0$ we define an \textbf{$n$-chain} $\chf $ to be a collection $\chf_{j_1,\hdots,j_n}$ of almost-local observables indexed by $(\Z^d)^n$ that
are
\begin{enumerate}[i)]
    \item Traceless:
    \begin{align}
        \overline{\tr}(\chf_{j_1,\hdots,j_n}) = 0 
    \end{align}
    \item Skew-adjoint:
    \begin{align}
        \chf_{j_1,\hdots,j_n}^* = -\chf_{j_1,\hdots,j_n}
    \end{align}
    \item Skew-symmetric:
    \begin{align}
        \chf_{j_1,\hdots,j_n} = (-1)^{|\sigma|}\chf_{j_{\sigma(1)},\hdots,j_{\sigma(n)}}
    \end{align}
    for any permutation $\sigma$ of $\{1,\hdots, n\}$,
    \item Uniformly localized: for any $\alpha \in \Z_{\ge 0}$ we have
    \begin{align}
        \sup_{j_1,\hdots, j_n\in \Z^d}\sup_{1\le i \le n} \|\chf_{j_1,\hdots,j_n}\|_{j_i, \alpha} < \infty. \label{eqn:chain-seminorm}
    \end{align}
\end{enumerate} 
\begin{remark}
Our grading convention is shifted by $1$ compared to Ref. \cite{LocalNoether}. 
\end{remark}
For $n> 0$ we define $C^{-n}$ as the Fr\'echet space of $n$-chains with seminorms given by (\ref{eqn:chain-seminorm}) for all $\alpha\ge 0$.
We extend this $n=0$ by letting $C^0=\mfkDal$. These form a (non-positively graded) cochain complex with the differential $\partial: C^{-n-1}\ra C^{-n }$ is defined for $n>0$ by 
\begin{equation}
    (\partial \derf)_{j_1,\ldots ,j_n}=\sum_{j_0\in\ZZ^d} \derf_{j_0,\ldots, j_n}.
\end{equation}
For $n=0$ it is defined by
\begin{equation}
    \partial\derf(A) =\sum_{j\in\ZZ^d}[\derf_j,A], 
\end{equation}
for any $A\in \mathscr{A}_{al}$. What's more, there is a graded Lie bracket $C^{-m}\times C^{-n}\ra C^{-m-n}$ which is defined for $m,n>0$ as 
\begin{align}
    \{\derf,\derg \}_{j_1,\hdots,j_{m+n}} := \sum_{\sigma}\frac{\sgn(\sigma)}{m!n!}[\derf_{j_{\sigma(1)},\hdots, j_{\sigma(m)}}, \derg_{j_{\sigma(m+1)},\hdots, j_{\sigma(m+n)}}],
\end{align}
where the sum is taken over all permutations on $m+n$ symbols.
For $m=0$ and $n>0$ we let
\begin{equation}
\{\derF,\derf\}_{j_1,\ldots,j_n}=[\derF, \derf_{j_1,\ldots,j_n}],
\end{equation}
while for $m=n=0$ we let $\{\derF,\derG\}=[\derF,\derG]$. The differential $\partial$ and the graded bracket $\{\cdot,\cdot\}$ make $C^\bullet=\oplus_{n=0}^\infty C^{-n}$ into a (non-positively graded) dg-Lie algebra.

We conclude this section with a description of the class of automorphisms of $\mathscr{A}_{al}$ obtained by integrating paths of UAL derivations. Let $\derF:\R\to \mfkDal$ be a smooth path of derivations, denoted $t \mapsto \derF_t$. By Prop. E.1 of \cite{LocalNoether}, for every $A\in\mathscr{A}_{al}$ there is a unique smooth path $t\mapsto A_t \in \mathscr{A}_{al}$ of observables satisfying $A_0=A$ and $\frac{dA_t}{dt}= \derF_t(A_t)$. For each $t\in \R$ this gives a map $A\mapsto A_t$ which is a $*$-automorphism of $\mathscr{A}_{al}$ that extends to a $*$-automorphism of $\mathscr{A}$. We denote it 
\begin{align}
     \tau\exp\left(\int_0^t \derF_tdt\right). \label{eqn:texp}
\end{align}
We call automorphisms obtained in this way \textbf{locally generated automorphisms} or LGAs for short. These have an action on $\mfkDal$ by conjugation which we write as 
\begin{align}
    \alpha(\derF) := \alpha\circ \derF \circ \alpha^{-1}.
\end{align}
This action can be promoted to an action on the entire complex $C^\bullet$ which commutes with the differential by allowing an LGA to act on an $n$-chain elementwise: $\alpha(\derf)_{j_1,\hdots,j_n} :=  \alpha(\derf_{j_1,\hdots,j_n})$. 

\subsection{States}\label{sec:families}

If $\psi$ is a state on $\mathscr{A}$ and $\alpha$ is an LGA, we write $\psi^{\alpha} := \psi\circ \alpha$. The fixed points of this action will play an important role in what follows. We say $\alpha$ preserves $\psi$ if $\psi^\alpha = \psi$. A derivation $\derF\in \mfkDal$ is said to preserve $\psi$ if $\psi(\derF(A))=0$ for all $A\in \mathscr{A}_{al}$. An observable $A\in\mathscr{A}_{al}$ is said to preserve $\psi$ if $\psi(\derF(A))=0$ for all $\derF\in\mfkDal$. We write $\mfkDal^\psi$ and $\SA_{al}^\psi$ for the derivations and observables that preserve $\psi$; these are closed subspaces of $\mfkDal$ and $\SA_{al}$ in their respective Fr\'echet topologies. Unitary elements of $\mathscr{A}_{al}^\psi$ satisfy the following properties:
\begin{lemma}\label{lem:local-unitary-preserving}
    If $V\in U(\mathscr{A}_{al})$ preserves $\psi$ then $\psi(V) \in U(1)$, and $\psi(VA) = \psi(V)\psi(A)$ for any $A\in \mathscr{A}_{al}$.
\end{lemma}
\begin{proof}
    Since $\mathscr{A}_{al}$ is norm-dense in $\mathscr{A}$, we have $\psi(VB) = \psi(BV)$ for any $B\in \mathscr{A}$.
    Let $A\in \mathscr{A}_{al}$ and let ($\mathcal{H}$, $\pi$, $\Omega$) be the GNS triple of $\psi$. Since $\psi$ is pure, $\pi$ is irreducible. In particular $\pi(\mathscr{A})$ is weakly dense in the bounded operators on $\mathcal{H}_\psi$, so there is a sequence $\{P_k\}_{k\in \Z_{\ge 0}}$ of elements of $\mathscr{A}$ such that $\pi(P_k)$ converges weakly to $P = |\Omega\rangle\langle\Omega|$, and we have
    \begin{align}
        \psi(VA) &= \lim_{k\to \infty}\psi(VAP_k) \nonumber \\
        &= \lim_{k\to \infty}\psi(AP_kV)\nonumber \\
        &= \psi(A)\psi(V),
    \end{align}
    which proves the second statement. The first follows from $\overline{\psi(V)}\psi(V) = \psi(V^*V) = 1$.
\end{proof}

The space $\mfkDal^\psi \subset \mfkDal$ is a Lie subalgebra which can be resolved to a dg-Lie subalgebra $C_\psi^\bullet \subset C^\bullet$. We put $C_\psi^0:=\mfkDal^\psi \subset \mfkDal$, while for $k>0$ we define $C_\psi^{-k}$ to be the (closed) subspace of $k$-chains $\derf$ for which $\derf_{j_1,\hdots,j_k}\in \mathscr{A}^\psi_{al}$ for all $j_1,\hdots, j_k \in \Z^d$. It is easy to see that $C_\psi^\bullet$ is preserved by the differential $\partial$ and is closed with respect to the bracket $\{\,,\}$.

We will be interested in several special classes of pure states. For states $\psi$,$\psi'$ on quasilocal  algebras $\SA$ and $\SA'$, we write $\psi\otimes \psi'$ for the resulting state on the quasilocal  algebra $\SA\otimes \SA'$.

\begin{definition}
    Let $\psi$ be a pure state of $\SA$. We say $\psi$ is
    \begin{enumerate}[i)]
        \item \textbf{factorized} if for each $j\in \Z^d$ there is a pure state $\psi_j$ on $\mathscr{A}_j$ such that $\psi|_{\mathscr{A}_{j}} = \psi_j$,
        \item \textbf{short-range entangled} (SRE for short) if there exists an LGA $\alpha$ such that $\psi\circ\alpha$ is a factorized pure state,
        \item \textbf{invertible} if there is another state $\psi'$ such that $\psi\otimes \psi'$ is SRE,
        \item \textbf{gapped} if there is a $\derH \in \mfkDal$ such that $\psi$ is a gapped ground state of $\derH$. That is, there exists $\Delta>0$ such that $-i \psi(A^* \derH(A))\geq \Delta (\psi(A^*A)-|\psi(A)|^2)$ for all $A\in\SAal$. 
    \end{enumerate}
\end{definition}
Factorized states model trivial systems. Short-range entangled states model systems in a trivial topological phase as they can be prepared from factorized states by a local Hamiltonian evolution \cite{QImeetsQM}. Invertible states model systems in invertible topological phases as introduced by A.Kitaev \cite{Kitaevlecture}. These are phases that have inverses, i.e., it is possible to stack the system with another system, such that the composite is in a trivial phase. 
\begin{prop}
Every invertible state is gapped.
\end{prop}
\begin{proof}
Given a state $\psi \otimes \psi'$ and an observable $A \in \SA \otimes \SA'$, one can define a partial average $\psi'(A) \in \SA$ that on observables of the form $\CO \otimes \CO' \in \SA \otimes \SA'$ is given by $\psi'(\CO \otimes \CO'):=  \psi'(\CO') \CO$. The value on general observables is obtained by linear extension, and it is a standard fact that the resulting map is a conditional expectation; in particular, we have $\|\psi'(A)\| \leq \|A\|$. For any $A \in \SAal$ and $r \geq 0$, we have
\begin{align}
\inf_{B \in \SA_{B_j(r)}} \|B - \psi'(A)\| \leq \inf_{B \in (\SA\otimes\SA')_{B_j(r)}} \|\psi'(B-A)\| \leq \inf_{B \in (\SA\otimes\SA')_{B_j(r)}} \|B-A\|,
\end{align}
and so partial averaging takes almost-local observables to almost-local observables. We also define partial averaging of a derivation $\derF\in \mfkDal$ by $\psi'(\derF)(A) = \psi'(\derF(A\otimes \boldsymbol{1}))$ for any $A \in \SAal$, which is again a derivation because $\psi'$ is a conditional expectation.

Let $(\psi',\SA')$ be an inverse of $(\psi,\SA)$, and let $\alpha$ be an LGA on the composite system $\SA\otimes\SA'$ such that $\Psi_0 := (\psi \otimes \psi') \circ \alpha$ is factorized. Let us choose a UAL derivation $\derF$ for the composite system such that $\Psi_0$ is a gapped ground state for $\derF$ with a gap greater than $\Delta > 0$ (we can choose $\derF$ to be $\partial \chf$ for an on-site $\chf \in C^1$). Then $(\psi \otimes \psi')$ is a gapped ground state for $\alpha(\derF)$. Let $\derH$ be a UAL derivation of $\SA$ obtained from $\alpha(\derF)$ by partial averaging over $\psi'$. Then for any $A \in \SAal$ we have
\begin{multline}
-i \psi(A^* \derH(A) ) = -i (\psi \otimes \psi')( \alpha( B^*) \alpha(\derF(B)) )  = -i \Psi_0( B^* \derF(B)) \geq \\ \geq \Delta \left(  \Psi_0(B^*B ) -|\Psi_0(B)|^2 \right) = \Delta \left(\psi( A^*A )-|\psi(A)|^2 \right)
\end{multline}
where $B = \alpha^{-1}(A \otimes 1)$. Thus, $\psi$ is a gapped ground state for $\derH$ with a gap greater than $\Delta > 0$.    
\end{proof}

It is believed that every 1d gapped state is invertible. For $d>1$ there are many examples of gapped states which are not invertible (for example, topologically ordered states).  

We now turn to the situation when the state is a smooth function of some parameter space $M$, which we take to be a smooth manifold. The set $C^\infty(M,C^\bullet)$ of smooth functions\footnote{That is, functions $f^k:M\to C^k$ for each $k\le 0$ such that in any set of smooth coordinates on $M$ all partial derivatives of $f^k$ exist and are continuous.} $M\to C^\bullet$ is a cochain complex valued in $C^\infty(M)$-modules. Suppose that $\{\psi_x\}_{x\in M}$ is a family of states parametrized by points of $M$. Then the set
\begin{align}
    C^\infty(M, C_\psi^\bullet) := \{f\in C^\infty(M, C^\bullet) : f(x) \in C_{\psi_x}^\bullet \ \forall x\in M\}
\end{align} 
is another cochain complex valued in $C^\infty(M)$-modules. For any $k>0$, we set $\Omega^{k}(M, C^\bullet) := \Hom_{C^\infty(M)}(\wedge^kTM,C^\infty(M, C^\bullet))$. Similarly, we define $\Omega^{k}(M, C_\psi^\bullet):= \Hom_{C^\infty(M)}(\wedge^kTM$,$C^\infty(M, C_\psi^\bullet))$.
\begin{definition}\label{def:smooth-gapped}
    Let $\psi = \{\psi_x\}_{x\in M}$ be family of states indexed by points of a smooth manifold $M$. We say $\psi$ is \textbf{smooth} if for every $A\in\mathscr{A}_{al}$ the function $x\mapsto \psi_x(A)$ is a smooth function on $M$, and there is a $\derG\in \Omega^1(M,\mfkDal)$ such that
        \begin{align}
            d\psi(A) = \psi(\derG(A)) \hspace{4mm} \forall A\in \mathscr{A}_{al}. \label{eqn:psiGparallel}
        \end{align}
\end{definition}
We say $\psi$ is \textbf{parallel} with respect to $\derG$ if (\ref{eqn:psiGparallel}) holds. We will sometimes write $(\psi,\derG)$ for a smooth family when we want to specify a particular $\derG$ with respect to which it is parallel.

We say a smooth family of states is factorized, SRE, or invertible if it is so pointwise. If $M$ is connected, then a smooth family $\psi$ is SRE or invertible if it is so for any one $x_0\in M$. This happens because two states connected by a smooth path $\gamma:[0,1]\to M$ are related by the LGA obtained by exponentiating (as in (\ref{eqn:texp})) $\derG$ along $\gamma$. 

We say a smooth family $\psi$ is gapped if there is a $\derH\in C^\infty(M,\mfkDal)$ such that $\psi_x$ is a gapped groundstate of $\derH(x)$ for each $x\in M$. If $M$ is connected, then $\psi$ is gapped iff it is gapped for any one $x_0\in M$ (this follows from a partition of unity argument, but we do not prove this here since we will not need this fact). Conversely, if $\psi$ is a family of states such that $\psi_x(A)$ is a smooth function of $x\in M$ for every $A\in \mathscr{A}_{al}$, and if there is a $\derH\in C^\infty(M,\mfkDal)$ such that $\psi_x$ is a gapped groundstate of $\derH(x)$ for each $x\in M$, then under some extra technical assumptions\footnote{See Assumption 1.2 in \cite{moon2020automorphic}.} $\psi$ can be shown to be smooth \cite{LocalNoether, moon2020automorphic}.

We conclude this section with a discussion of smooth families of LGAs. 
\begin{definition}
    A family of LGAs $\alpha = \{\alpha_x\}_{x\in M}$ is \textbf{smooth} if for every $A\in \mathscr{A}_{al}$ the map $M \to \mathscr{A}_{al}$ taking $x\mapsto \alpha_x(A)$ is smooth, and there is a $\derG\in \Omega^1(M,\mfkDal)$ such that
    \begin{align}
        d\alpha(A) = \alpha(\derG(A)) \hspace{4mm} \forall A\in \mathscr{A}_{al}.
    \end{align}
    If such a $\derG$ exists, it is unique, and we denote it by $\alpha^{-1}d\alpha$.
\end{definition}
The most natural way to produce an LGA is to integrate a $\mfkDal$-valued 1-form along a path, as in (\ref{eqn:texp}). As we will see below, this can be extended coherently to the setting of smooth families of LGAs. Let $M$ be a manifold and $I = [0,r]$ an interval. Define the vertical complex $\Omega^\bullet(M\times I,C^\bullet)_{vert}$ as the quotient of $\Omega^\bullet(M\times I,C^\bullet)$ by the set of elements $\mathsf{a}$ for which $\iota_{\frac{\partial}{\partial \theta}}\mathsf{a}= 0$, and we write $\mathsf{a}_{vert}$ for the image of $\mathsf{a}$ under the projection to $\Omega^\bullet(M\times I,C^\bullet)_{vert}$, followed by the obvious inclusion back into $\Omega^\bullet(M\times I,C^\bullet)$. Write $j_s:M\to M\times I$ for the inclusion as $M\times \{s\}$.
\begin{prop}\label{lem: path-ordered-exp}
    Let $M$ be a manifold and let $\derG \in \Omega^1(M\times I, \mfkDal)$. Then there is a unique smooth family of LGAs $\alpha$ on $M\times I$ satisfying $\alpha\circ j_0 = \boldsymbol{1}$ and
    \begin{align}
        \frac{d}{ds}\alpha(A) = \alpha(\iota_{\frac{\partial}{\partial s}}\derG(A)) \hspace{4mm} \forall A\in \mathscr{A}_{al}, \label{eqn:lem:path-ordered-exp}
    \end{align}
    where $s$ is the coordinate on $I$.
\end{prop}
\begin{proof}
    Let $\gamma_x:I\to M$ be the function taking $s\mapsto (x,s)$. By \cite{LocalNoether} Prop. E.1, for each $x\in M$ there is a unique solution to (\ref{eqn:lem:path-ordered-exp}) with $\alpha\circ j_0=\boldsymbol{1}$ which we denote by $\alpha_{(x,s)}$
    By \cite{LocalNoether} Proposition E.2 we have
    \begin{align}
        (\alpha^{-1}d\alpha)_{(x,t)} = -\int_0^t\alpha_{(x,t)}^{-1}\circ \alpha_{(x,s)}(\iota_{\frac{\partial}{\partial s}}d(\derG_{vert})_{(x,s)})ds + \derG_{vert} \ \  \in \Omega^1(M\times I,\mfkDal). \label{eqn:ada}
    \end{align}
\end{proof}
\begin{definition}
    Given $\derG\in \Omega^1(M\times I, \mfkDal)$ and $\alpha$ as in Proposition \ref{lem: path-ordered-exp}, we denote  $\tau\exp(\int_0\derG):= \alpha$, while for $s\in I$ we write $\tau\exp(\int_0^s\derG):= \alpha\circ j_s$.
\end{definition}
Notice that $\tau\exp\int_0\derG$ depends only on the vertical component $\derG_{vert}$ of $\derG$. We close this section with a description of a gauge action of smooth families of LGA on on smooth families of states (the proof is straightforward and is omitted):
\begin{prop}\label{prop:gauge-states}
    If $(\psi,\derG)$ is a smooth family of states on $M$ and $\alpha$ is a smooth family of LGAs on $M$, then $\psi^\alpha = \psi\circ\alpha$ is parallel with respect to 
    \begin{align}
        \derG^{\alpha} &:= \alpha^{-1}(\derG) +\alpha^{-1}d\alpha.
    \end{align}
\end{prop}

\subsection{Higher Berry curvatures and classes}\label{sec:berrydef}
We can now state the main result of \cite{LocalNoether}, which allows for the construction of the invariants which are the subject of this paper. Recall that a cochain complex $(K^\bullet,\partial)$ is nullhomotopic if there is a map $h:K^\bullet\to K^{\bullet-1}$ (which we call a contracting homotopy) satisfying $h\circ \partial + \partial \circ h = \boldsymbol{1}$.
\begin{theorem}\label{thm:acyclic}
    Let $M$ be a smooth manifold.
    \begin{enumerate}[i)]
        \item The cochain complex $C^\bullet$ is nullhomotopic via a contracting homotopy $h:C^\bullet\to C^{\bullet-1}$. For any $k\ge 0$ the unique $C^\infty(M)$-linear extension $h:\Omega^k(M,C^\bullet)\to \Omega^k(M,C^{\bullet-1})$ is also a contracting homotopy.
        \item Suppose $\psi$ is a smooth gapped family of states. Then for every $k\ge0$ the complex $\Omega^k(M,C_\psi^\bullet)$ is nullhomotopic via a $C^\infty(M)$-linear contracting homotopy $h^\psi: \Omega^k(M,C_\psi^\bullet) \to \Omega^k(M,C_\psi^{\bullet-1})$.
    \end{enumerate}
\end{theorem}

We can extend the graded bracket on $C^\bullet$ to $\Omega^\bullet(M,C^\bullet)$ by defining the bracket between $\mathsf{a} \in \Omega^p(M,C^k)$ and $\mathsf{b} \in \Omega^q(M,C^\ell)$ as
\begin{align}
    \{\mathsf{a},\mathsf{b}\}(X_1,\hdots, X_{p+q}) = (-1)^{kq}\sum_{\sigma} \frac{\sgn{\sigma}}{p!q!}\{\mathsf{a}(X_{\sigma(1)}, \hdots X_{\sigma(p)}),
    \mathsf{b}(X_{\sigma(p+1)}, \hdots X_{\sigma(p+q)})\}
\end{align}
for any vector fields $X_1,\hdots, X_{p+q}$. The differentials $d$ and $\partial$ on $\Omega^\bullet(M)$ and $C^\bullet$ extend naturally to $\Omega^\bullet(M,C^\bullet)$, and we get a total differential which acts on $\mathsf{a} \in \Omega^p(M,C^k)$ by 
\begin{align}
    \boldsymbol{d}(\mathsf{a}) := d(\mathsf{a}) + (-1)^p\partial(\mathsf{a}).\label{eqn:totaldiff}
\end{align} This graded bracket and total differential make $\Tot(\Omega^\bullet(M,C^\bullet))$ into a dg-Lie algebra over $C^\infty(M)$, with $\Tot(\Omega^\bullet(M,C_\psi^\bullet))$ as a dg-Lie subalgebra.

Note that our sign conventions differ from those in \cite{LocalNoether} -- in particular our $\partial$ and $d$ commute instead of anticommuting\footnote{Explicitly, what we call $C^q$ is called $\mathcal{N}^q$ in \cite{LocalNoether}, and what we call $\partial:C^q \to C^{q+1}$ is $(-1)^{q}\partial : \mathcal{N}^q \to \mathcal{N}^{q+1}$ in the notation of \cite{LocalNoether}.}. This choice will make the explicit calculations in Section \ref{sec:thoulessquant} easier. The price we pay is that although $d$ is a graded derivation of $\Tot(\Omega^\bullet(M,C^\bullet))$, $\partial$ is not. Instead, $(-1)^{p}\partial: \Omega^p(C^q)\to \Omega^p(C^{q+1})$ is. 

The element $\derG\in\Omega^1(M,\mfkDal)$ can be interpreted as a connection 1-form on the trivial graded bundle $M\times C^\bullet$. Its covariant derivative is the graded derivation $D$ of $\Tot(\Omega^\bullet(M,C^\bullet))$ given by $d+\{\derG, \cdot \}$.
Its curvature $\derF\in \Omega^2(M,\mfkDpal)$  satisfies 
\begin{align}
    D\circ D(A) = \{\derF, A\} \hspace{3mm} \forall A \in \Omega^\bullet(M,C^\bullet)
\end{align}
and is given by the usual formula $\derF = d\derG + \frac{1}{2}\{\derG,\derG\}$.

The higher Berry invariants are constructed by solving the following Maurer-Cartan equation. Recall that $\boldsymbol{d}$ is the total differential on $\Tot(\Omega^\bullet(M,C^\bullet))$, given by (\ref{eqn:totaldiff}).
\begin{prop}\label{thm:MC}
    Suppose that $(\psi,\derG)$ is a gapped smooth family. Then there exists an element $\derg^\bullet\in \Tot^1(\Omega^\bullet(M,C^\bullet))$, whose component in $\Omega^{n+1}(M,C^{-n})$ we denote\footnote{What we call $\derg^{(n)}$ was called $g^{(n-1)}$ in \cite{LocalNoether}.} $\derg^{(n)}$, that satisfies $\derg^{(0)}=\derG$ and
        \begin{align}
        \boldsymbol{d}\derg^\bullet + \frac{1}{2}\{\derg^\bullet,\derg^\bullet\} = 0. \label{eqn:descent}
    \end{align}
    Furthermore we have $\derg^{(n)}\in \Omega^{n+1}(M,C_\psi^{-n})$ for all $n>0$.
\end{prop}
We call (\ref{eqn:descent}) the Maurer-Cartan equation or alternatively the descent equation. The proof of Prop. \ref{thm:MC} in \cite{LocalNoether} involves writing out eq. (\ref{eqn:descent}) as a system of equations for $\derg^{(n)}$
\begin{align}
    d\mathsf{g}^{(n-1)} + (-1)^n\partial \mathsf{g}^{(n)} + \frac{1}{2}\sum_{k=0}^{n-1}\{\mathsf{g}^{(k)},\mathsf{g}^{(n-k-1)}\} = 0. \label{eqn:descent-explicit}
\end{align}
and solving (\ref{eqn:descent-explicit}) successively for $n=1,2,\ldots,$ using the exactness of the bi-complex $\Omega^{\bullet}(M,C^{-\bullet}_\psi)$ with respect to $\partial$ in positive degrees. For this reason, $\derg^{(n+1)}$ will be called the descendant of $\derg^{(n)}$. Notice that (\ref{eqn:descent-explicit}), together with the fact that $\derg^{(n)}\in \Omega^{n+1}(M,C_\psi^{-n})$ for $n>0$, imply that $d\psi(\derg^{(n)}) = \psi(\partial \derg^{(n+1)})$. 

In Section \ref{sec:restictions} we introduce an operation we call ``evaluating against the origin'' (equation (\ref{eval-fundamental})). The evaluation of $\derg^{(d+1)}$ at the origin is an element of $\Omega^{d+2}(M,\mathscr{A}_{al})$ denoted by $\angles{\derg^{(d+1)},[*]}$. Evaluating $\derg^{(d+1)}$ against the origin and applying the family of states $\psi$ to this observable-valued form we obtain a $(d+2)$-form on $M$:
\begin{align}
    \omega^{(d+2)} := \psi(\angles{\derg^{(d+1)},[*]}) \in \Omega^{d+2}(M,\C)
\end{align}
which is closed because
\begin{align*}
    -d\omega^{(d+2)} &= (-1)^{d+2}\psi(\angles{\partial \derg^{(d+2)}, [*]})+\frac12 \sum_{k=1}^{d}\psi\big(\angles{\{\derg^{(k)},\derg^{(d-k+1)}\}, [*]}\big).
\end{align*}
The first term above vanishes because $\angles{\derh,[*]}=0$ wheneverif $\derh$ is $\partial$-exact and the second term vanishes because $\psi(\angles{\derh, [*]})=0$ if $\psi(\derh_{j_1,\hdots, j_{d+1}})=0$ for all $j_1,\hdots,j_{d+1}\in\Z^d$.

\begin{definition}\label{def:higher-berry}
Suppose $(\psi,\derG)$ is a smooth family of states such that a solution $\derG^\bullet$ of the MC equation (\ref{eqn:descent}) exists. Then the cohomology class $[\omega^{(d+2)}]\in H_{dR}^{d+2}(M,i\RR)$ is independent of the choice of $\derG$ and $\derG^\bullet$. It is called the \textbf{higher Berry class} of the smooth family $\psi$. 
\end{definition}
By Proposition \ref{thm:MC} all gapped smooth families (and thus all SRE and invertible smooth families) have a solution to (\ref{eqn:descent}) and thus a higher Berry class. 

For $d=0$, this is just the usual Berry curvature (Chern number) of a line bundle. For $d=1$ the higher Berry curvature is a closed 3-form $\omega^{(3)}$ whose cohomology class measures the flow of Berry curvature from the right half of the spin chain to the left half \cite{HermeleEtAl}. 


When the system under consideration is equipped with an on-site action of a compact Lie group $G$ we can consider equivariant smooth families $(\psi,\derG)$ parametrized by a $G$-manifold\footnote{The 1-form $\derG$ is assumed to be $G$-equivariant too.}  and there is an equivariant version of the above descent procedure. We describe it here in the case that $G=U(1)$ and the $U(1)$-action on the parameter space $M$ is trivial -- in other words, $(\psi,\derG)$ is a smooth family of states on $M$ that is $U(1)$-invariant for each $x\in M$, and $\derG$ is a $U(1)$-invariant element of $\Omega^1(M,\mfkDal)$.

The generators of the on-site $U(1)$ actions form a 1-chain $\derq^{(1)}$ such that the derivation $\derQ:= \partial \derq^{(1)}$ preserves $\psi_x$ for every $x\in M$. 
Consider the Cartan complex $\Omega^{\bullet, \bullet}(M,C^\bullet)_{U(1)}:= \Omega^\bullet(M)\otimes S^\bullet\R \hat\otimes (C^\bullet)_{U(1)}$, where $S^\bullet \R$ is the algebra of polynomials $\R[t]$ on one generator $t$, which we assign degree 2, and $(C^\bullet)_{U(1)}$ are the $U(1)$-invariant chains. If $\psi$ is a $U(1)$-invariant smooth family of states, we define $\Omega^{\bullet, \bullet}(M,C_\psi^\bullet)_{U(1)}$ similarly.

The equivariant descent equation reads
\begin{align}
    \boldsymbol{d}\derg^\bullet + \frac{1}{2}\{\derg^\bullet ,\derg^\bullet \} + \derQ\otimes t = 0 \label{eqn:equivariant-descent},
\end{align}
where $\derg^\bullet\in \Tot^1(\Omega^{\bullet, \bullet}(M,C^\bullet)_{U(1)})$ and the component of $\derg$ in $\Omega^{1,0}(M,C^0)$ is $\derG$.
As before, the components of $\derg^\bullet$ in $\Omega^{\bullet, \bullet}(M,C^{d+1})_{U(1)}$ can be evaluated against the origin (as in (\ref{eval-fundamental})) to produce closed forms on $M$ which altogether form a class in the equivariant cohomology $H_{U(1)}^{d+2}(M)$, and it can be shown that this cohomology class is independent of the choice of solution $\derg^\bullet$ of the equivariant Maurer-Cartan equation.

The components of $\derg$ in various degrees encode various cohomology invariants one can assign to a $U(1)$-invariant smooth family of states. Letting $\mathsf{g}^{(n,k)}$ be the component of $\Omega^{n+1-2k,2k}(M,C^{-n})_{U(1)}$, the first few components of $\derg^\bullet$ can be arranged in the following table
\begin{align}
    \begin{array}{ccc} \label{eqn:invariants-table}
        \mathsf{g}^{(0,0)} &  \ & \ \\
        \mathsf{g}^{(1,0)} &  \mathsf{g}^{(1,1)} & \ \\
        \mathsf{g}^{(2,0)} &  \mathsf{g}^{(2,1)} & \ \\
        \mathsf{g}^{(3,0)} &  \mathsf{g}^{(3,1)} & \mathsf{g}^{(3,2)}
    \end{array}
\end{align}
Each component $\mathsf{g}^{(n,k)}$ associates a closed $n+1-2k$-form to a smooth family of  $U(1)$-invariant states in $n-1$ dimensions. 
For instance, for a family of $1$-dimensional $U(1)$-invariant states $\psi(\angles{\mathsf{g}^{(2,0)}, [*]}) = \omega^{(3)}\in\Omega^3(M)$ is the higher Berry curvature discussed above. For a family of $0$-dimensional $U(1)$-invariant states $\psi(\angles{\mathsf{g}^{(1,0)}, [*]})\in\Omega^2(M)$ is the ordinary Berry curvature. In general, the invariant corresponding to $\mathsf{g}^{(n+1,k)}$ is the descendant of the one corresponding to $\mathsf{g}^{(n,k)}$. 

The rest of the terms in (\ref{eqn:invariants-table}) represent the following invariants. As we just described, the first column contains the Berry curvature $\derg^{(1,0)}$ and its descendants\footnote{One might wonder if $\derg^{(0,0)}$ corresponds to some invariant. It should associate a 1-form to a family of (-1)-dimensional states. If one interprets a (-1)-dimensional state as a phase, then $\derg^{(0,0)}$ is nothing but the pullback of the form $d\theta$ on $U(1)$.}. At the top of the second column is the $U(1)$ charge $\mathsf{g}^{(1,1)}$ (which can be thought of as giving a locally constant function on a family of 0d systems). The descendant $\mathsf{g}^{(2,1)}$ of charge gives the usual Thouless pump for 1d systems, while $\mathsf{g}^{(n,1)}$ for $n>2$ give the higher Thouless pump invariants. Finally $\mathsf{g}^{(3,2)}$ is the Hall conductance (whose descendants $\derg^{(n,2)}$ for $n\ge 4$ are not pictured).

The main results of this work deal with the entries $\mathsf{g}^{(2,0)}$ (1d Berry curvature) and $\mathsf{g}^{(3,1)}$ (2d Thouless pump) in the above table. The proof of quantization of the 2d Thouless invariant will hinge on showing that these two are related by the process of inserting a flux at infinity. We remark that this pattern holds more generally. We will not treat these rigorously in this work, but let us simply state a few other instances of this ``diagonal" relationship. Beginning with a single 2-dimensional $U(1)$-invariant state we obtain an $S^1$-family by inserting a $\theta$-domain wall at the $x$-axis for every $\theta\in U(1)$, and the charge pumped along the domain wall as one cycles $\theta\in S^1$ from $0$ to $2\pi$, which can be interpreted as the charge of a fluxon, can be shown to equal the Hall conductance. This is the original Laughlin argument and relates $\derg^{(3,2)}$ (Hall conductance) to $\derg^{(2,1)}$ (1d Thouless pump). By inserting another domain wall, along the $y$-axis this time, we obtain a $S^1\times S^1$-family of states whose ordinary Berry curvature (given by $\derg^{(1,0)}$) is again the Hall conductance: this is the basis of the proof of the Hall conductance quantization \cite{hastingsmichalakis}.

This paper is only concerned with the first two columns of (\ref{eqn:invariants-table}), so we will use a simplified notation for their entries. Instead of $\derg^{(n,0)}$ we will simply write $\derg^{(n)}$, while $\derg^{(n,1)}$ will be denoted $\dert^{(n)}$. Then $\derg^{(n)}$ satisfy the ordinary descent equations (\ref{eqn:descent-explicit}), of which we will need only the first two:
\begin{align}
    \partial \derg^{(1)} &= \derF \nonumber \\
    \partial \derg^{(2)} &= -D\derg^{(1)}.
\end{align}
Meanwhile the first three descent equations for $\dert^{(n)}$ are:
\begin{align}
    \partial \dert^{(1)} &= \derQ \nonumber \\ 
    \partial \dert^{(2)} &= -D\dert^{(1)}  \nonumber \\
    \partial \dert^{(3)} &= D\dert^{(2)} + \{\dert^{(1)}, \derg^{(1)}\}. \label{eqn:descent-thouless}
\end{align}

\section{Localization properties}
In this section we introduce a few tools to deal with localization properties of chains, derivations, and automorphisms. 
We begin in Section \ref{sec:localized} by defining the notion of a derivation, chain, or automorphism that is confined on a given region in $\Z^d$, then in Section \ref{sec:restictions} we discuss some ways to produce derivations confined on a given region.
\subsection{Confined maps}\label{sec:localized}
\begin{definition}
    A linear map $F: \mathscr{A}_{al}\to \mathscr{A}_{al}$ is $h$-\textbf{confined} on a region $X\subset \Z^d$ if for every finite $Y\subset \Z^d$ it satisfies
    \begin{align}
        \|F(B)\| \le \sum_{z\in X}h(d(z,Y))\|B\|
    \end{align}
    for all $B$ localized on $Y$. If we omit $h$ and say $F$ is confined on $X$, we mean that it is $h$-confined on $X$ for some function $h$ with superpolynomial decay.
\end{definition}

\begin{definition}
    For $n>0$, an $n$-chain $\derf$ is $h$-\textbf{confined} on a region $X\subset \Lambda$ if 
    \begin{align}
        \|\derf_{j_1,\hdots,j_n}\|\le \min_{k=1,\hdots,n} h(d(X,j_k)).
    \end{align} As before we say $\derf$ is confined on $X$ if it is h-confined on $X$ for some superpolynomially decreasing $h$.
\end{definition}
For $\derf \in \Omega^\bullet(M,C^\bullet)$, being pointwise confined on a region $X\subset \Z^d$ is rarely a sufficiently strong condition -- one must impose some kind of uniformity on the decay function $h$: we say $\derf\in C^\infty(\R^n,C^\bullet)$ is \textbf{smoothly confined} on $X$ if for any multi-index $\mu\ge 0$ and any $x\in \R^n$, there is a neighbourhood $V\ni x$ and a superpolynomially decaying function $h$ such that $\partial^\mu\derf(x)$ is $h$-confined on $X$ for all  $x\in V$. Since this is a local property we extend this definition to $\derf\in \Omega^\bullet(M,C^\bullet)$ for a manifold $M$ by requiring $\derf$ to be smoothly confined on $X$ in any chart. Finally, a chain-valued form $\derf\in \Omega^k(M,C^\bullet)$ is defined to be smoothly confined on $X$ if $\derf(\sigma)$ is smoothly confined on $X$ for any multivector field $\sigma$. Lastly, if $\derf,\derg,\derh \in \Omega^\bullet(M,C^\bullet)$ we say $\derf$ \textbf{smoothly interpolates} between $\derg$ on $X$ and $\derh$ on $X^c$ if $\derf-\derg$ is smoothly confined on $X^c$ and $\derf-\derh$ is smoothly confined on $X$.

The notion of confinement is compatible with many operations typically applied to chains, as summarized in Proposition \ref{prop:localization-properties} below, whose proof appears in Appendix \ref{sec:appendix-confined}.

For $X,X'\subset \Lambda$, we say $Y$ is a \textbf{stable intersection} of $X$ and $X'$ if there exist a $c>0$ such that $X(r)\cap X'(r) \subset Y(cr)$ for all $r\ge0$ [note: this is not the only definition, and maybe there is a better one]. Note in particular that the origin is a stable intersection of the set $\{x\ge 0\} \subset \Z$ and its complement, and the positive $y$-axis is a stable intersection of the upper half-space in $\Z^2$ with the $y$-axis.

\begin{prop}\label{prop:localization-properties}
    Let $\mathsf{a},\mathsf{b} \in \Omega^\bullet(M,C^\bullet)$. Let $X,X'$ be any subsets of $\Z^d$ and let $Y$ be a stable intersection of $X$ and $X'$.
    \begin{enumerate}[i)]
        \item If $\mathsf{a}$ is smoothly confined on $X$ then $\partial \mathsf{a}$ is smoothly confined on $X$.
        \item If $\mathsf{a}$ is smoothly confined on both $X$ and $X'$ then it is smoothly confined on $Y$.
        \item If $\mathsf{a}$ and $\mathsf{b}$ are smoothly confined on $X$ and $X'$, respectively, then $\{\mathsf{a},\mathsf{b}\}$ is smoothly confined on $Y$.
    \end{enumerate}
\end{prop}

LGAs also have some desirable localization properties if they are produced by integrating a confined derivation:
\begin{prop}\label{prop:localization-LGAs}
    Let $\derG\in \Omega^1(M\times [0,1], \mfkDal)$ and suppose $\derG_{vert}$ is smoothly confined on $X\subset \Z^d$. Let $\alpha := \tau\exp\int_0\derG$. Then $\alpha^{-1}d\alpha$ is smoothly confined on $X$, and $\alpha(\derF)-\derF$ is smoothly confined on $X$ for any $\derF \in \Omega^\bullet(M,\mfkDal)$.
\end{prop}

We call an element of $\derG\in \mfkDal$ \textbf{inner} if it is inner as a derivation of $\mathscr{A}_{al}$, in other words there is some $A\in \mathscr{A}_{al}$ such that $\derG(B)=[A,B]$ for all $B\in\mathscr{A}_{al}$ (we say $A$ is \textbf{associated} to $\derG$). Since $\mathscr{A}_{al}$ has trivial center, any two elements of $\mathscr{A}_{al}$ associated to the same $\derG$ are related by a multiple of $\boldsymbol{1}$. 
The notion of confined derivations allows an explicit description of the set of inner UAL derivations: a derivation $\derG \in \mfkDal$ is inner iff it is confined on a bounded $X\subset \Z^d$. This is a consequence of the more general statement:

\begin{prop}\label{prop:smooth-lift}
    Let $M$ be a smooth manifold and suppose $\derF \in \Omega^\bullet(M,\mfkDal)$. Then $\derF$ is of the form $\derF = \ad_{A}$ for some antiselfadjoint $A\in \Omega^\bullet(M,\mathscr{A}_{al})$ iff it is confined on a bounded region of $\Z^d$.
\end{prop}

Let $\psi$ be a state on $\mathscr{A}$. If $\derF \in \mfkDal$ is inner we may unambiguously define $\psi(\derF)$ as $\psi(A)-\overline{\tr}(A)$ for any $A\in \mathscr{A}_{al}$ associated to $\derF$. This procedure is covariant with respect to automorphisms of $\mathscr{A}$. Indeed, since $\overline{\tr}$ is the unique tracial state on $\mathscr{A}$\footnote{Indeed any finite-dimensional matrix algebra has a unique tracial state, and $\mathscr{A}$ is a norm-limit of these.}, we have $\overline{\tr}^\alpha = \overline{\tr}$ for any automorphism of $\mathscr{A}$. Thus
\begin{equation}
    \psi^\alpha(\derF) = \psi^\alpha(A) - \overline{\tr}(A)
    = \psi(\alpha(A)) - \overline{\tr}(\alpha(A))
    = \psi(\alpha(\derF))
\end{equation}
for any $A\in \mathscr{A}_{al}$ associated to $\derF$.

\subsection{Restricting and evaluating chains}\label{sec:restictions}
Given a subset $X\subset\Z^d$ and an $n$-chain $f$ define the \textbf{restriction} of $f$ to $X$ as the $n$-chain $\res_{X}(\derf)$ given by
\begin{align}
    \res_{X}(\derf)_{j_1,\hdots,j_n}= \left\{\begin{array}{cc}
        \derf_{j_1,\hdots,j_n} & \text{ if } j_i\in X \text{ for every } i=1,\hdots ,n  \\
        0 & \text{ otherwise} 
    \end{array} \right.
\end{align} It is easy to see that $\res_X$ commutes with $\derF$ for any $\derF\in\mfkDal$, and that for $k>0$ and a smooth manifold $M$ the obvious extension of $\res_X$ to $\Omega^\bullet(M,C^k)$ commutes with the exterior derivative $d$. Notice also that $\res_X(\derf)$ is confined on $X$, and if $\derf$ is confined on $Y\subset \Z^d$ then $\res_X(\derf)$ is also.

Suppose $\derh$ is an $n$-chain and let $X \subset \Z^d$. We may form the $(n-1)$-chain $\partial\res_X \derh - \res_X\partial \derh$, or $[\partial,\res_X]\derh$ for short. This chain measures the current of the quantity $\derh$ across the boundary of the region $X$. Since $[\partial, \res_X] = -[\partial, \res_{X^c}]$, it is clear that this $(n-1)$-chain is confined on both $X$ and $X^c$ (and thus, by Proposition \ref{prop:localization-properties}, on any stable intersection of $X$ and $X^c$). In what follows, we will most often set $X$ to be one of the half-spaces $\mathbb{H}_i := \{ (x_1,\hdots, x_n): x_i\le 0 \}\subset \Z^d$. 

We end this section by establishing some notation which will be useful throughout the rest of this work. If $\derh^{(2)}$ is a 2-chain that is confined on a region which has bounded stable intersection with $\partial\mathbb{H}_i = \{ (x_1,\hdots, x_n): x_i = 0 \}\subset \Z^d$, then we define
\begin{align}
    \angles{\derh^{(2)}, [\partial \mathbb{H}_i]} := \sum_{j\in \Z^d}([\partial, \res_{\mathbb{H}_i}]\derh^{(2)})_j\in \mathscr{A}_{al} \label{eval-hyperplane}
\end{align}
the sum on the right-hand side being absolutely convergent in $\mathscr{A}_{al}$.

Now suppose $\derh^{(d+1)}$ is any $d+1$-chain. Then $[\partial, \res_{\mathbb{H}_d}]...[\partial, \res_{\mathbb{H}_1}]\derh^{(d+1)}$ is a 1-chain that measures the $d$-dimensional circulation of $\derh$ around the origin\footnote{In the language of \cite{LocalNoether} this is the same as contracting with the conical partition $\{X_k\}_{k=1}^{d+1}$ with $X_k = \mathbb{H}_k\backslash (\mathbb{H}_1\cup\hdots \cup \mathbb{H}_{k-1})$ for $1\le k\le d$ and $X_{d+1} = \mathbb{H}_1^c\cap...\cap \mathbb{H}_{d}^c$.}, and we define
\begin{align}
    \angles{\derh^{(d+1)}, [*]} := \sum_{j\in \Z^d} ([\partial, \res_{\mathbb{H}_d}]...[\partial, \res_{\mathbb{H}_1}]\derh^{(d+1)})_j \in \mathscr{A}_{al}, \label{eval-fundamental}
\end{align}
where the sum is again absolutely convergent. Notice that the observable $\angles{\derh, [*]}$ is traceless and associated to the inner derivation $\partial[\partial, \res_{\mathbb{H}_d}]...[\partial, \res_{\mathbb{H}_1}]\derh$. Since $\partial$ commutes with $[\partial, \res_{X}]$ for any $X$, it follows that $\angles{\derh, [*]}=0$ if $\derh$ is $\partial$-closed.

\section{1d higher Berry quantization} \label{sec:berryquant}
We are now ready to prove our first main result: that for invertible families the higher Berry class has an integral refinement. Recall that the exponential exact sequence $0\to 2\pi i \Z \to i \RR \to U(1) \to 0$ gives rise to an isomorphism $H^2(M,\bUM) \cong H^3(M,2\pi i \Z)$, where $\bUM$ is the sheaf of continuous $U(1)$-valued functions on $M$ \cite{WZWreview,BrylLoop}. Let $\iota: H^2(M,\bUM) \hookrightarrow H^3(M,i\RR)$ denote the composition of this isomorphism with the usual \v{C}ech-de Rham map. For a gapped smooth family $\psi$ of 1d states, let $\omega^{(3)}=\psi(\langle\derg^{(2)},[*]\rangle)\in \Omega^3(M,i\R)$ denote its higher Berry curvature.

\begin{theorem}\label{thm:1dBerryquantization}
    To any smooth family $(\psi,\derG)$ of invertible 1d states on $M$ we can associate a class $[h] \in H^2(M,\bUM)$ such that $\iota([h]) = [\omega^{(3)}]$. In particular, the class $[\omega^{(3)}/2\pi i]\in H^3(M,\RR)$ is integral.
\end{theorem}
We will first prove the result when $\psi$ is SRE, then extend the result to the case when $\psi$ is invertible. Let $\{U_a\}_{a\in J}$ be an open cover of $M$ such that for any $a_1,\hdots, a_n\in J$ the intersection $U_{a_1}\cap \hdots \cap U_{a_n}$ is either empty or contractible. Write $U_{ab}:=U_a\cap U_b$ and $U_{abc} := U_a\cap U_b \cap U_c$. 

The proof will proceed by constructing a Deligne-Beilinson cocycle whose curvature is $\omega^{(3)}$. Recall \cite{WZWreview,BrylLoop} that a Deligne-Beilinson 2-cocycle is a triple $(h_{abc} \in C^{\infty}(U_{abc},U(1)), a_{ab} \in \Omega^1(U_{ab}, i \RR), b_a \in \Omega^2(U_a,i \RR))$ such that
\begin{align}
h_{abc}h_{acd} &= h_{abd}h_{bcd}\\
h_{abc}^{-1} d h_{abc} &= a_{ab} - a_{ac} + a_{bc}\\ 
d a_{ab} &= b_a - b_b.
\end{align}
In the physics literature, such 2-cocycles are called 2-form gauge fields and define a connection on a line bundle gerbe over $M$ \cite{murray2007introduction}. The existence of such a cocycle with $d b_a= \omega^{(3)}|_{U_a}$ implies quantization of $[\omega^{(3)}/2 \pi i]$ \cite{WZWreview,BrylLoop}.


Pick a basepoint $x_0\in M$. For each $a\in J$ let $H_a: U_a\times [0,1] \to M$ be a smooth homotopy from the constant map $U_a\to \{x_0\}$ to the identity map $U_a \to U_a$. Let $L= \Z_{\le 0} \subset \Z$ and $R = L^c$.

For $a\in J$, define $\tilde\alpha_a^{-1} := \tau \exp (\int_0^1H_a^*\derG)$\footnote{Here and below, for a map $f:M \to N$ between smooth manifolds $M,N$ and a differential form $\derA \in \Omega^{\bullet}(N)$ we write $f^* \derA \in \Omega^{\bullet}(M)$ for the pullback of differential forms.}. This is a smooth family of automorphisms on $U_a$ that provides a local trivialization of $\psi$ in the sense that for every $x\in U_a$ we have $\psi_x = \psi_{x_0}\circ (\tilde{\alpha}_a)_x^{-1}$. Next, define $\alpha_a^{-1} := \tau \exp (\int_0^1 H_a^*\partial\res_Lh\derG)$ (where $h$ is the contracting homotopy from Theorem \ref{thm:acyclic} $i)$), which can be thought of as a restriction of $\tilde\alpha_a^{-1}$ to the left half-line. Define $\alpha_{ab}:=\alpha_a\circ\alpha_b^{-1}$. 

The family of states $\psi\circ \alpha_{ab}^{-1}$ differs from $\psi$ appreciably only near the origin. In fact, there exists a smooth family of unitaries $V_{ab}\in C^\infty(U_{ab}, \mathscr{A}_{al})$ satisfying $\psi\circ \alpha_{ab}^{-1} = \psi \circ \Ad_{V_{ab}^{-1}}$ and $\overline{\tr}(V_{ab}^{-1}dV_{ab}) = 0$. To see this, note  first that we have $\psi\circ \alpha_{ab}^{-1} = \psi \circ \tilde\alpha_a\circ\tilde{\alpha}_b^{-1}\circ\alpha_{ab}^{-1}$. It is clear that $\alpha_{ab}^{-1}$ is of the form $\tau\exp(\int_0^1\derH)$ for some $\derH$ confined on $L$. On the other hand, we have $\tilde\alpha^{-1}_a\circ\alpha_a = \tau\exp(\alpha_a^{-1}(\partial\res_Rh(\derG)))$ and $\tilde\alpha_a\circ \tilde\alpha_b^{-1}\circ \alpha_b\circ \tilde\alpha_a^{-1} = \tau\exp(\tilde\alpha_a\circ\alpha_b^{-1}(\partial\res_Rh(\derG)))$. Since both of these are of the form $\tau\exp(\int_0^1\derH)$ for a $\derH$ confined on $R$ their product is also of this form. Thus we may use Lemma \ref{lem:alunitary} to guarantee that such a $V_{ab}$ exists.

Define $W_{abc} := V_{ac}^{-1}\alpha_{ab}(V_{bc})V_{ab}$, which is a smooth $U(\mathscr{A}_{al})$-valued function satisfying $\overline{\tr}(W_{abc}^{-1}dW_{abc}) = 0$. We also define LGAs $\beta_{ab}:=\Ad_{V_{ab}^{-1}}\circ \alpha_{ab}$ and the derivation-valued forms $\derB_{ab}:=\beta^{-1}_{ab} d\beta_{ab}$. It is easy to check that $\beta_{ab}$ and $\derB_{ab}$ preserve $\psi$, and that $\beta_{ab}$ and $W_{abc}$ satisfy the following two relations:
\begin{align}
    \beta_{ab}\circ\beta_{bc}\circ\beta_{ac}^{-1} &= \Ad_{W_{abc}}^{-1},\label{eqn:Wabc-beta}\\
    W_{abd}^{-1}W_{acd}W_{abc} &= \beta_{ab}(W_{bcd}). \label{eqn:Wabc-cocycle}
\end{align}
From the first equation above it follows that $W_{abc}$ preserves $\psi$, so $h_{abc}:= \psi(W_{abc})$ is a smooth $U(1)$-valued function on $U_{abc}$. From the second it follows that $h_{abc}h_{acd}=h_{abd}h_{bcd}$, i.e., $h_{abc}$ is a cocycle (both this and the previous statements use Lemma \ref{lem:local-unitary-preserving}, which we will continue to use throughout).

Suppose $V_{ab}'$ is another smooth choice of unitaries satisfying $\psi\circ\alpha_{ab}^{-1}:=\psi\circ \Ad_{{V_{ab}'}^{-1}}$. Then $Y_{ab} := V_{ab}^{-1}V_{ab}'$ preserves $\psi$, so $g_{ab}:= \psi(Y_{ab})$ is a smooth $U(1)$-valued function on $U_{ab}$. Writing $W_{abc}':= {V'_{ac}}^{-1}\alpha_{ab}(V_{bc}')V_{ab}'$, we have
\begin{align}
    W_{abc}' = Y_{ac}^{-1}W_{abc}\beta_{ab}(Y_{bc})Y_{ab},
\end{align}
and so $\psi(W_{abc}') = g_{ab}g_{bc}g_{ac}^{-1}\psi(W_{abc})$. Thus the 2-cocycle constructed from $Y_{ab}$ differs from the 2-cocycle $h_{abc}$ constructed from $V_{ab}$ by a 2-coboundary, and so $h_{abc}$ defines an element $[h_{abc}]\in \check{H}^2(M,\bUM)\cong H^3(M,\ZZ)$ which is independent of the choice of $V_{ab}$'s. 

Differentiating (\ref{eqn:Wabc-beta}) gives
\begin{align}
    \ad_{W_{abc}^{-1}dW_{abc}} = \Ad_{W_{abc}}^{-1}d\Ad_{W_{abc}} = \derB_{ac} - \derB_{bc} - \beta_{bc}^{-1}(\derB_{ab}).
\end{align}
Since $W_{abc}^{-1}dW_{abc}$ is traceless, this implies that
\begin{align}
    h_{abc}^{-1}dh_{abc}=\psi\left(\derB_{ac}-\derB_{bc}-\beta_{bc}^{-1}(\derB_{ab})\right), \label{eqn:dh_abc}
\end{align}
where on the right-hand side we are evaluating a state on an inner derivation as in Section \ref{sec:localized}.

Below we write $\derF_{\derH}:=d\derH + \frac{1}{2}\{\derH,\derH\}$ for any $\derH\in \Omega^1(M,\mfkDal)$.
Since $\psi=\psi_0\circ \tilde\alpha_a^{-1}$, the family $\psi$ is parallel with respect to $\tilde\alpha_a d\tilde\alpha_a^{-1}$. By Lemma \ref{lem:connection-interpolation} there is a $\derC_a\in\Omega^1(U_a,\mfkDal)$ such that $\psi$ is parallel with respect to $\derC_a$, $\derC_a$ smoothly interpolates between $\tilde\alpha_a d\tilde\alpha_a^{-1}$ on $L$ and $\derG$ on $R$, and $\derF_{\derC_a}$ smoothly interpolates between $0$ and $\derF$.
Now define
\begin{align}
    a_{ab} &:= \psi(\derB_{ab} - \derC_b + {\beta_{ab}}^{-1}(\derC_a)), \label{eqn:ab} \\
    b_{a} &:= \psi(d\derC_a + \frac{1}{2}\{\derC_a,\derC_a\} - \partial \res_R\derg^{(1)}) \label{eqn:ab2}
\end{align}
Using Proposition \ref{prop:localization-LGAs} one can check that $\derB_{ab} - \derC_b + {\beta_{ab}}^{-1}(\derC_a)$ is smoothly confined on both $L$ and $R$, ensuring that $a_{ab}$ is well-defined and smooth. A similar argument shows this for $b_a$ as well.
\begin{lemma}
    We have
\begin{align}
    h_{abc}^{-1}dh_{abc} =  a_{ab} -  a_{ac} +  a_{bc}.
\end{align}
\begin{proof}
    Using the fact that $\psi\circ\beta_{bc}^{-1}= \psi$ we have
    \begin{align}
        a_{ab} -  a_{ac} +  a_{bc} &= \psi(\beta_{bc}^{-1}(\derB_{ab}) - \derB_{ac}+ \derB_{bc})) + \psi(\beta_{bc}^{-1}(\beta_{ab}^{-1}(\derC_a))-\beta_{ac}^{-1}(\derC_a))
        \label{eqn:dh2}
    \end{align}
    The second term equals $\psi(\Ad_{W_{abc}}(\derC_a)-\derC_a) = 0$, and so (\ref{eqn:dh2}) agrees with the expression (\ref{eqn:dh_abc}).
\end{proof}
\end{lemma}
\begin{lemma}
    We have
    \begin{align}
        da_{ab} = b_a-b_b \label{eqn:da_ab}
    \end{align}
\end{lemma}
\begin{proof}
On an overlap $U_{ab}$, $\psi$ is parallel with respect to both $\derC_a$ and $\derC_b$. by Proposition \ref{prop:gauge-states} it is parallel with respect to $(\derC_b)^{\beta_{ab}}$. Thus it is parallel with respect to $\frac{1}{2}((\derC_a)^{\beta_{ab}} + \derC_b) = \frac{1}{2}(\beta_{ab}^{-1}(\derC_a) + \derB^{ab} + \derC_b)$. Using this we get
\begin{align}
    da_{ab} &= \psi(\derB_{ab} + \frac{1}{2}\{\derB_{ab},\derB_{ab}\} - d\derC_b - \frac{1}{2}\{\derC_b,\derC_b\} + {\beta_{ab}}^{-1}(d\derC_a + \frac{1}{2}\{\derC_a,\derC_a\}))\nonumber\\
    &= \psi(-\derF_{\derC_b} + {\beta_{ab}}^{-1}(\derF_{\derC_a})),
\end{align}
where as before we write $\derF_{\derC_a}:=d\derC_a+\frac{1}{2}\{\derC_a,\derC_a\}$, and we used the fact that $\derF_{\derB_{ab}}=0$.
To split the last line into two terms, we regularize $-\derF_{\derC_b} + {\beta_{ab}}^{-1}(\derF_{\derC_a})$ by adding $\partial \res_R\derg^{(1)} - {\beta_{ab}}^{-1}(\partial \res_R \derg^{(1)})$. Since $\res_R\derg^{(1)} - {\beta_{ab}}^{-1}(\res_R\derg^{(1)})$ is a 1-chain confined at the origin, each of whose terms is traceless and has zero expectation under $\psi$, we have $\psi(\partial \res_R\derg^{(1)} - {\beta_{ab}}^{-1}(\partial \res_R \derg^{(1)})) = 0$. Thus:
\begin{align}
     da_{ab} &= \psi(-\derF_{\derC_b} + {\beta_{ab}}^{-1}(\derF_{\derC_a}) - \partial \res_R\derg^{(1)} + {\beta_{ab}}^{-1}(\partial \res_R\derg^{(1)}))\nonumber\\
     &= b_a - b_b.
\end{align}
\end{proof}
\begin{lemma}
    $db_a = -\psi(\langle\derg^{(2)},[*]\rangle)$
\end{lemma}
\begin{proof}
\begin{align}
    db_a &= \psi(D_{\derC_a}\derF_{\derC_a} + \{\derG-\derC_a,\derF_{\derC_a}\} + D_{\derG}\partial \derg^{(1)}_R) \nonumber \\
    &=\psi(\{\derG-\derC_a,\derF_{\derC_a}\}) + \psi(\partial [\partial, \res_{R}]\derg^{(2)}).\label{eqn:lastone}
\end{align}
The first term in (\ref{eqn:lastone}) is zero since $\{\derG-\derC_a,\derF_{\derC_a}\}$ is associated to the absolutely convergent $\sum_{j\in \Z^d}\derF_{C_a}(h(\derG-\derC_a)_j))$ and $\derF_{\derC_a}$ preserves $\psi$. The result then follows from the fact that $\angles{\derg^{(2)},[*]}$ is associated to $\partial[\partial,\res_L]\derg^{(2)} = -\partial[\partial,\res_R]\derg^{(2)}$.
\end{proof}
This concludes the proof of Theorem \ref{thm:1dBerryquantization} in the case that the family $\psi$ is SRE.

Although so far we defined the refined higher Berry class $[h_{abc}]\in H^3(M,2\pi i\ZZ)$ only for SRE families, it is easy to extend the definition as well as the proof of Theorem \ref{thm:1dBerryquantization}  to arbitrary smooth invertible families. Let $\psi$ be such a family. Pick $x_0\in M$, let $\psi'_{x_0}$ be an inverse for $\psi_{x_0}$, and consider the family $\psi\otimes \psi'_{x_0}=\{\psi_x\otimes \psi'_{x_0} \}_{x\in M}$. This family is SRE at the point $x_0\in M$ and thus it is SRE on the whole $M$. Define the refined higher Berry class $[h_{abc}]$ of $\psi$ as that of $\psi\otimes \psi'_{x_0}$. It is independent of the choice of $\psi'_{x_0}$. Indeed, suppose  $\psi''_{x_0}$ is another inverse for $\psi_{x_0}$. Since the SRE families $\psi_{x_0}\otimes\psi'_{x_0}$ and $\psi_{x_0}\otimes\psi''_{x_0}$ are constant, their refined higher Berry class vanish. Further, it is easy to see that for any two smooth SRE families $\psi,\psi'$ the refined higher Berry class of $\psi\otimes\psi'$ is the sum of the refined higher Berry classes of $\psi$ and $\psi'$. Therefore the refined higher Berry class of the SRE family $\psi\otimes\psi'_{x_0}\otimes \psi''_{x_0}\otimes\psi_{x_0}$ is equal to both the refined higher Berry class of $\psi\otimes\psi'_{x_0}$ and the refined higher Berry class of $\psi\otimes \psi''_{x_0}$.

\textbf{Example:} Let us describe an example of the family of SRE states  for which the class $[\omega^{(3)}/2 \pi i] \in H^3(M,\RR)$ is non-trivial. It is essentially the example from \cite{HermeleEtAl} adapted for our setting. Let $(\chi \in [0,\pi], \theta \in [0,\pi], \phi \in [0,2 \pi))$ be spherical coordinates on $S^3$ with $(\theta,\phi)$ being spherical coordinates on $S^2$ at fixed $0<\chi<\pi$. The equatorial $S^2$ is located at $\chi = \pi/2$. The regions $\chi\leq \pi/2$ and $\chi \geq \pi/2$ correspond to the upper $S^3_+$ and the lower hemisphere $S^3_-$, respectively. Let $B:S^2 \to \CB(\mathbb{C}^2)$ defined by $B(\theta,\phi) = \vec{n}(\theta,\phi) \cdot \vec{\sigma}$ with $\vec{n}$ being a unit vector in $\RR^3$ pointed in the direction $(\theta,\phi) \in S^2$ and $\vec{\sigma} = (\sigma^x, \sigma^y, \sigma^z)$ being Pauli matrices, and let $p_{\pm} = (1\pm B)/2$. Let us choose $\epsilon \in (0,\pi/4)$. Let us choose a smooth family of rank one projections $P^{(+)}:S_+^3 \to \CB(\mathbb{C}^2) \otimes \CB(\mathbb{C}^2)$ such that in the neighbourhood $0 \leq \chi \leq \epsilon$ it is some constant rank one projection, while in the neighbourhood $(\pi/2-\epsilon) \leq \chi \leq \pi/2$ it is given by 
$$
P^{(+)}(\theta,\phi,\chi) = p_+(\theta,\phi) \otimes p_-(\theta,\phi).
$$
Similarly, we define a smooth family of rank one projections $P^{(-)}:S_-^3 \to \CB(\mathbb{C}^2) \otimes \CB(\mathbb{C}^2)$ such in the neighbourhood $(\pi-\epsilon) \leq \chi \leq \pi$ it is some constant rank one projection, while in the neighbourhood $\pi/2 \leq \chi \leq (\pi/2+\epsilon)$ it is given by 
$$
P^{(-)}(\theta,\phi,\chi) = p_-(\theta,\phi) \otimes p_+(\theta,\phi).
$$

Let us consider a one-dimensional lattice system with $\SA_j \cong \CB(\mathbb{C}^2)$. For $k \in \ZZ$, we let $P^{(+)}_k \in \SAal$ be a smooth family of local observables on $S^3_+$ corresponding to $P^{(+)}$ under the isomorphism $\SA_{2k} \otimes \SA_{2k+1} \cong \CB(\mathbb{C}^2) \otimes \CB(\mathbb{C}^2)$. Similarly, let  $P^{(-)}_k \in \SAal$ be a smooth family of observables on $S^3_-$ corresponding to $P^{(-)}$ via the isomorphism $\SA_{2k-1} \otimes \SA_{2k} \cong \CB(\mathbb{C}^2) \otimes \CB(\mathbb{C}^2)$. We let $\{\psi_x\}_{x \in S_+^3}$ be a family of pure states of $\SA$ uniquely defined by the requirement that when restricted to $\SA_{2k}\otimes\SA_{2k+1}$ it is given by $A\mapsto {\rm Tr} P^{(+)}_k A$. This is a smooth family of states on $S_+^3$ which is parallel with respect to $\derG'^{(+)} \in \Omega^1(S_+^3,\mfkDal)$
given by 
\begin{equation}
    \derG'^{(+)} = \sum_{k \in \ZZ} [P^{(+)}_k,d P^{(+)}_k].
\end{equation}
Similarly, we consider a family of pure states $\{\psi_x\}_{x \in S_-^3}$ of $\SA$ defined by the requirement that when restricted to $\SA_{2k-1}\otimes\SA_{2k}$ it is given by $A\mapsto {\rm Tr} P^{(-)}A$. This is a smooth family of states on $S_-^3$ which is parallel with respect to $\derG'^{(-)} \in \Omega^1(S_+^3,\mfkDal)$
given by 
\begin{equation}
    \derG'^{(-)} = \sum_{k \in \ZZ} [P^{(-)}_k,d P^{(-)}_k].
\end{equation}
It is easy to see that the two families of states agree on the equatorial $S^2$. Moreover, the resulting family of states on the whole $S^3$ is smooth.  To see this, consider an open neighbourhood $E$ of the equatorial $S^2$ given by $\pi/2-\epsilon/2<\chi<\pi/2+\epsilon/2.$ $\psi\vert_E$ is a family of  product states whose restriction to $\SA_{2k}$ (resp. $\SA_{2k+1}$) is given by $A\mapsto {\rm Tr} p_+ A$ (resp.  $A\mapsto {\rm Tr} p_- A$ ). Therefore $\psi\vert_E$ is parallel with respect to a derivation-valued 1-form $\derA \in \Omega^2(E,\mfkDal)$ given by the sum of on-site terms $[p_+,d p_+]$ on $\SA_{2k}$ and $[p_-,d p_-]$ on $\SA_{2k+1}$ for $k \in \ZZ$. Since $S_+^3,S_-^3,$ and $E$ form an open cover of $S^3$, there exists a partition of unity $1=\rho_++\rho_-+\rho_E$, where $\rho_+,\rho_-,\rho_E$ are smooth functions supported on $S_+^3,S_-^3,E$, respectively. We get a globally-defined 1-form $\derG\in \Omega^1(S^3,\mfkDal)$ such that $\psi$ is parallel with respect to $D=d+\derG$ by letting $\derG=\rho_+ \derG'^{(+)}+\rho_-\derG'^{(-)} +\rho_E \derA$.


For such $\derG$, we can choose $\derg^{(1)} \in \Omega^2(S^3,C^{-1}_{\psi})$ such that 1) $\derg^{(1)}_{2k}, \derg^{(1)}_{2k+1} \in \SA_{2k} \otimes \SA_{2k+1}$ on $S^3_+$; 2) $\derg^{(1)}_{2k}, \derg^{(1)}_{2k-1} \in \SA_{2k-1} \otimes \SA_{2k}$ on $S^3_-$; 3) on $E$, $\derg^{(1)}_{2k}$ and $\derg^{(1)}_{2k+1}$ are given by $\frac12 [d p_+, d p_+] \in \SA_{2k}$ and $\frac12 [d p_-, d p_-] \in \SA_{2k+1}$, respectively. Note that for any $k$, $D(\derg^{(1)}_{2k}+ \derg^{(1)}_{2k+1})=0$ on $S^3_+$ and $D(\derg^{(1)}_{2k-1}+ \derg^{(1)}_{2k})=0$ on $S^3_-$. Therefore $\angles{\mathsf{g}^{(2)}, [*]}$ vanishes on $S^3_-$ and coincides with $D \derg^{(1)}_0$ on $S^3_+$. The integral of the higher Berry curvature over $S^3$ is
\begin{multline}
\int_{S^3} \psi(\angles{\mathsf{g}^{(2)}, [*]}) = \int_{S^3_+} \psi(D \derg^{(1)}_0) = \int_{S^2 = \partial S^3_+} \psi(\derg^{(1)}_0) = \int_{S^2} \frac{1}{2} \text{Tr} \left( p_+ [d p_+, d p_+] \right) = 2 \pi i.
\end{multline}


\section{2d Thouless pump}\label{sec:thoulessquant}
Let $\psi_M$ be a smooth gapped family of 2d $U(1)$-invariant states over a compact manifold $M$. The 2d Thouless pump invariant $\angles{\dert^{(3)}, [*]}$ associates to it a class in de Rham cohomology $H^2(M,{i}\RR)$. 
In this section we show, using a proof analogous to Laughlin's flux insertion argument, that this class can be refined to a class in integral cohomology. Here is a roadmap of the argument. 

We extend $\psi_M$ to a family of states $\psi$ on $M\times S^1$ by performing a $U(1)$ gauge transformation on the right half-plane, which we interpret as implementing a $2\pi$-flux insertion at a point at infinity on the $y$-axis $\partial \mathbb{H}_1$. Although the 3-form $\psi(\angles{\derg^{(1)},[\partial \mathbb{H}_2]})$ which measures the Berry curvature flux across the $x$-axis is divergent, its vertical component $\psi(\angles{\derg^{(1)}_{vert},[\partial \mathbb{H}_2]})$ is finite because varying the $S^1$-parameter $\theta \in [0, 2\pi)$ changes $\psi$ appreciably only near the $y$-axis in $\ZZ^2$. In Section \ref{sec:thouless-as-berry} below we compute this form and show that it equals $d\theta\wedge \angles{\dert^{(3)}, [*]}$. Thus in particular the total Berry curvature pumped across the $x$-axis during the flux insertion is equal to $2\pi$ times the 2d Thouless invariant. This result is true for any gapped U(1)-invariant smooth family $\psi_M$.

Then, in Section \ref{sec:2dThoulessQuant} we use an argument analogous to the proof of Berry flux quantization in Section \ref{sec:berryquant} to show that if $\psi_M$ is invertible then the total Berry curvature transported across the $x$-axis along the flux insertion, which is given by the fiber integral $\int_{S^1}\psi(\angles{\derg^{(1)}_{vert},[\partial \mathbb{H}_2]})$ and thus equals $2\pi\angles{\dert^{(3)}, [*]}$, is integral.

\subsection{2d Thouless pump as  a higher Berry curvature}\label{sec:thouless-as-berry}
Throughout this section we will use the following action of the de Rham complex $\Omega^\bullet(M)$ on $\Omega^\bullet(M,C^\bullet)$: for $\eta\in \Omega^p(M)$ and $\mathsf{a} \in \Omega^q(M,C^k)$ we put
\begin{align}
    \eta\wedge \mathsf{a}(X_1,\hdots, X_{p+q}) := \sum_{\sigma}\frac{\sgn(\sigma)}{p!q!}\eta(X_{\sigma(1)}, \hdots X_{\sigma(p)})
    \mathsf{a}(X_{\sigma(p+1)}, \hdots X_{\sigma(p+q)}).
\end{align}
for any vector fields $X_1, \hdots, X_{p+q}$.
It is easy to check that $d(\eta \wedge \mathsf{a}) = d\eta \wedge \mathsf{a} + (-1)^p\eta\wedge d\mathsf{a}$ and if $\mathsf{b} \in \Omega^r(M,C^\ell)$ then $\{\mathsf{b}, \eta\wedge \mathsf{a}\} = (-1)^{p(r+\ell)}\eta\wedge \{\mathsf{b},\mathsf{a}\}$.

Let $(\psi_M,\derG_M)$ be a gapped $U(1)$-invariant smooth family of 2d states on $M$. View it as a family of states on $M\times S^1$ that is constant in the $S^1$ direction, which we will also call $(\psi_M,\derG_M)$. We will often refer to a (chain-valued) differential form on $M$ and its pullback by the projection $M\times S^1\to M$ by the same symbol; it should be clear by context which is meant. 

Define $\rho = \tau\exp(\int_0\partial \res_{\mathbb{H}_1} \derq^{(1)}d\theta)$, where $\derq^{(1)}$ is the 1-chain consisting of the generators of the onsite $U(1)$ action, and $\theta$ is the coordinate on $S^1$. This is a smooth family of automorphisms on $M\times S^1$. Define
$\psi := \psi_M\circ \rho^{-1}$. This is a $U(1)$-invariant family of gapped states which represents the threading of a flux ``at infinity'' into the original family $\psi_M$ (see Figure \ref{fig:psirho} below).
\begin{figure}[ht]
\centering
\includegraphics[width=0.5\textwidth]{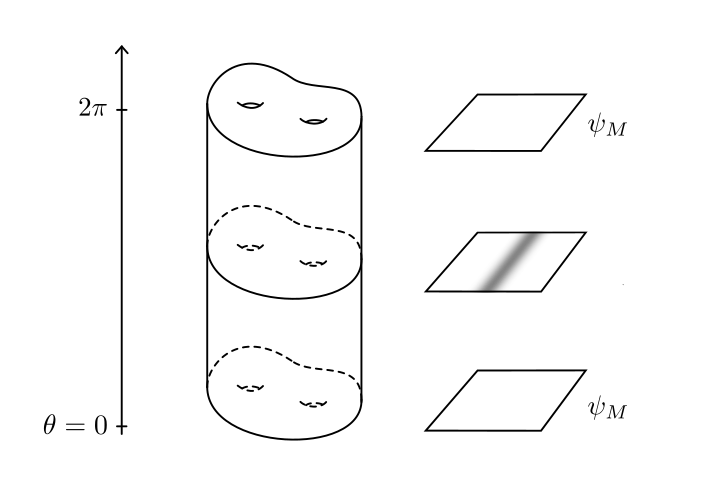}
\caption{A schematic depiction of $\psi := \psi_M\circ{\rho^{-1}}$. The shaded areas on the right side indicate regions in $\Z^2$ where $\psi$ differs from $\psi_M$. For $x\in M$ and $0\le \theta\le 2\pi$, the state $\psi_{(x,\theta)}$ is the state $(\psi_M)_x$ with a $\theta$-domain wall on the y-axis. Cycling $\theta$ from $0$ to $2\pi$ performs a flux insertion at infinity. }
\label{fig:psirho}
\end{figure}
By Proposition \ref{prop:gauge-states} $\psi$ is parallel with respect to the $U(1)$-invariant connection $\derG_M^{\rho^{-1}} = \rho(\derG_M - \partial\res_{\mathbb{H}_1} \derq^{(1)}d\theta)$. However, we will use a slightly different connection. Define $\derG \in \Omega^1(M\times S^1,\mfkDal)$ by
\begin{align}
    \derG :=  \rho(\derG_M - \partial\res_{\mathbb{H}_1}(\derq^{(1)} - \dert^{(1)})d\theta ). \label{eqn:G-fluxinsert}
\end{align}
where $\dert^{(1)}:= h^{\psi_M}(\derQ)$, with $h^{\psi_M}$ the contracting homotopy from Theorem \ref{thm:acyclic} $ii)$.
This differs from $\derG_M^{\rho^{-1}}$ by the term $\rho(\partial\res_{\mathbb{H}_{1}}\dert^{(1)})d\theta$ which preserves $\psi$ and is $U(1)$-invariant, so $\psi$ is still parallel with respect to $\derG$, and $\derG$ is still $U(1)$-invariant. The reason we choose $\derG$ instead of $\derG_M^{\rho^{-1}}$ is that its vertical component is confined on the $y$-axis $\partial\mathbb{H}_1$. In what follows, we will need to choose all our derivation-valued forms to satisfy this constraint. 

\begin{theorem}\label{thm:thouless-as-berry}
    The smooth gapped family $(\psi,\derG)$ admits a solution to the MC equation (\ref{eqn:descent}) such that $\derg^{\bullet}_{vert}$ is smoothly confined on the $y$-axis $\partial\mathbb{H}$, and 
    \begin{align}
        \psi(\angles{\derg^{(2)}_{vert},[\partial \mathbb{H}_2]}) = d\theta \wedge \psi_M(\angles{\dert^{(3)}}, [*])
    \end{align}
    where $\dert^{(3)}$ is a solution to the equivariant Maurer-Cartan equation (\ref{eqn:equivariant-descent}) for $(\psi_M,\derG_M)$.
\end{theorem}
\begin{proof}
We proceed as though we were computing the 1d Berry invariant for the family $\psi$ by solving the MC equation for $\derG$. At each step this will require adding a counterterm to ensure that the vertical component of $\derg^{(k)}$ is confined on $\partial\mathbb{H}_1$. As it turns out, choosing these counterterms will precisely involve solving the equivariant descent equations (\ref{eqn:descent-thouless}). Indeed, the first counterterm $d\theta \wedge \rho(\partial\res_{\mathbb{H}_{1}}\dert^{(1)})$, which was required to regularize the vertical component of $\derG$, already involved solving $\partial \dert^{(1)} = \derQ$.

Let $\derg^{\bullet}_M$ be a solution of the Maurer-Cartan equation (\ref{eqn:descent}) for $(\psi_M,\derG_M)$. Below we write $D_M = d + \{\derG_M, \cdot \}$ and $\derF_{M} = d\derG_M + \frac{1}{2}\{\derG_M,\derG_M\}$. We begin by computing the curvature of $\derG$:
\begin{align}
    \derF = \rho(\derF_M + D_M(d\theta\wedge \partial\res_{\mathbb{H}_1} \dert^{(1)})).
\end{align} The first step in the descent equation is to find a $\derg^{(1)}\in \Omega^2(M\times S^1,C_\psi^{-1})$ with $\partial \derg^{(1)} = \derF$. We will choose a $\derg^{(1)}$ of the form
\begin{align}
    \derg^{(1)} = \rho(\derg^{(1)}_M + d\theta\wedge \derf^{(1)}),\label{eqn:berry-curv-reg}
\end{align} where $\partial \derf^{(1)} = -D_M(\partial\res_{\mathbb{H}_1} \dert^{(1)})$.
Since $-D_M(\partial\res_{\mathbb{H}_1}\dert^{(1)}) = -\partial D_M\res_{\mathbb{H}_1}\dert^{(1)}$, we could use the naive expression
\begin{align}
    \derf_{naive}^{(1)} :=-D_M(\res_{\mathbb{H}_1} \dert^{(1)}).
\end{align}
However, since $\derf_{naive}^{(1)}$ isn't confined on $\partial\mathbb{H}_1$, it needs to be regularized. To do this, let $\dert^{(2)} \in \Omega^{1}(M,C_{\psi_M}^{-2})$ be a $U(1)$-invariant chain with $\partial \dert^{(2)} = -D_M\dert^{(1)}$, and set
\begin{align}
    \derf^{(1)} &:= -[\partial, \res_{\mathbb{H}_1}]\dert^{(2)}.
\end{align}
For the next step in the descent procedure, we seek $\derg^{(2)} \in \Omega^3(M\times S^1,C_\psi^{-2})$ which satisfies $\partial \derg^{(2)} = -D\derg^{(1)}$. Calculating $D\derg^{(1)}$ gives
\begin{align}
    D\derg^{(1)} &= \rho\bigg(D_M\derg_M^{(1)} + d\theta\wedge \big(\{\partial\res_{\mathbb{H}_1} \dert^{(1)}, \derg_M^{(1)}\} - D_M(\derf^{(1)})\big)\bigg),
\end{align}
so we choose the following ansatz for $\derg^{(2)}$:
\begin{align}
    \derg^{(2)} = \rho(\derg_M^{(2)} + d\theta\wedge \derf^{(2)}), \label{eqn:g2def}
\end{align}
where $\derf^{(2)}\in \Omega^2(M\times S^1, C_\psi^{-2})$ must satisfy
\begin{align}
    \partial \derf^{(2)} &=  -\{\partial\res_{\mathbb{H}_1} \dert^{(1)}, \derg_M^{(1)}\} + D_M(\derf^{(1)})\nonumber\\
    &= -\partial\bigg(\{\res_{\mathbb{H}_1}\dert^{(1)}, \derg^{(1)}_M\} + \res_{\mathbb{H}_1}D_{M}\dert^{(2)}\bigg).\label{eqn:x1}
\end{align}
This gives an unregularized expression for $\derf^{(2)}$:
\begin{align}
    \derf_{naive}^{(2)} &= -\{\res_{\mathbb{H}_1}\dert^{(1)}, \derg^{(1)}_M\} - \res_{\mathbb{H}_1}D_M\dert^{(2)}\nonumber\\
    &= -\{\res_{\mathbb{H}_1}\dert^{(1)}, \res_{\mathbb{H}_1^c} \derg^{(1)}_M\} - \res_{\mathbb{H}_1}\big(\{\dert^{(1)}, \derg^{(1)}_M\}+ D_M\dert^{(2)}\big). \label{eqn:x2}
\end{align}
Only the second term in (\ref{eqn:x2}) needs regularization, giving
\begin{align}
    \derf^{(2)} = -\{\res_{\mathbb{H}_1}\dert^{(1)},\res_{\mathbb{H}_1^c} \derg^{(1)}_M\} + [\partial, \res_{\mathbb{H}_1}]\dert^{(3)} 
\end{align}
where $\dert^{(3)}\in \Omega^2(M,C_{\psi_M}^{-3})$ is chosen such that
\begin{align}
    \partial\dert^{(3)} = \{\dert^{(1)}, \derg^{(1)}_M\}+ D_M\dert^{(2)}.
\end{align}
Clearly $d\theta\wedge \rho(\derf^{(2)})$ is the vertical component of $\derg^{(2)}$ and we have 
\begin{align}
    \psi(\angles{d\theta\wedge \rho(\derf^{(2)}),[\partial\mathbb{H}_2]}) &= d\theta \wedge \psi_M(\partial[\partial,\res_{\mathbb{H}_2}]\derf^{(2)}) \nonumber\\
    &= d\theta \wedge \psi_M(\partial[\partial,\res_{\mathbb{H}_2}][\partial, \res_{\mathbb{H}_1}]\dert^{(3)}))\nonumber\\
    &= d\theta \wedge  \psi_M(\angles{\dert^{(3)},[*]}).
\end{align}
    The first equality holds because $\psi_M = \psi\circ\rho$ and the second because $\psi_M(\angles{\{\res_{\mathbb{H}_1}\dert^{(1)},\res_{\mathbb{H}_1^c}\mathsf{g}_M^{(1)}\},[\partial\mathbb{H}_2]}) = 0$, since both $\dert^{(1)}$ and $\derg^{(1)}_M$ preserve $\psi_M$.
\end{proof}
\subsection{2d Thouless pump quantization}\label{sec:2dThoulessQuant}
    Having expressed the 2d Thouless invariant as the Berry curvature transport during flux insertion, we proceed to proving that this invariant is quantized if $\psi_M$ is SRE. The proof is along the lines of the proof of ordinary Berry curvature quantization in Section \ref{sec:berryquant}, but some modifications must be made because $\psi$ is not a family of 1d states -- instead it is only 1-dimensional in the $\theta$ direction. Let $\eta^{(2)}:=\psi_M(\angles{\dert^{(3)},[*]})$ be the 2d Thouless invariant of a family of states $\psi_M$ on $M$, and recall that we write $\iota: H^2(M,\bUM)\cong H^2(M,2\pi i \Z) \hookrightarrow H^3(M,i\RR)$ for the \v{C}ech-de Rham map.
    
    \begin{theorem}\label{thm:2dThoulessquantization}
    To any smooth family $(\psi_M,\derG_M)$ of invertible $U(1)$-invariant 2d states on $M$ we can associate a class $[h_{ab}] \in H^1(M,\bUM)$ such that $\iota([h_{ab}]) = 2\pi[\eta^{(2)}]$. In particular, the class $-i[\eta^{(2)}]\in H^2(M,\RR)$ is integral.
    \end{theorem}
    Before we begin, we will need the following Lemma:
    \begin{lemma}\label{lem:dergvert-confined}
        Let $\derG$ and $\derG_M$ be as in Section \ref{sec:thouless-as-berry}. Then $\derG-\derG_M$ is smoothly confined on $\partial \mathbb{H}_1$. In particular, $\derG_{vert}$ is smoothly confined on $\partial \mathbb{H}_1$.
    \end{lemma}
    \begin{proof}
        We have
        \begin{align}
            \derG - \derG_M = \rho(\derG) - \derG_M + d\theta\wedge\rho(\partial\res_{\mathbb{H}_1}(\derq^{(1)}) - \dert^{(1)}).
        \end{align}
        By $U(1)$-invariance of $\derG_M$ we have
        \begin{align}
            \derG - \derG_M = \tau\exp(\int_0^{2\pi}\partial\res_{\mathbb{H}_1}\derq^{(1)})(\derG_M) - \derG_M = \tau\exp(\int_0^{2\pi}-\partial\res_{\mathbb{H}_1^c}\derq^{(1)})(\derG_M) - \derG_M,
        \end{align}
        so by Proposition \ref{prop:localization-LGAs}, $\derG-\derG_M$ is smoothly confined on both $\mathbb{H}_1$ and $\mathbb{H}_1^c$, and thus it's confined on $\partial\mathbb{H}_1$.
        Next, since $\partial\derq^{(1)} = \derQ = \partial\dert^{(1)}$, there is a $\derk^{(2)}$ satisfying $\partial \derk^{(2)} = \derq^{(1)} - \dert^{(1)}$, and we have
        \begin{align}
            d\theta\wedge\rho(\partial\res_{\mathbb{H}_1}(\derq^{(1)} - \dert^{(1)})) = -d\theta\wedge\rho(\partial[\partial,\res_{\mathbb{H}_1}]\derk^{(2)})
        \end{align}
        which is confined on $\partial\mathbb{H}_1$ by the results of Section \ref{sec:restictions}.
    \end{proof}
    Let us introduce a key ingredient of the proof of Theorem \ref{thm:2dThoulessquantization}. Define the following smooth families of LGAs on $M\times [0,2\pi]$:
    \begin{align}
        \tilde\gamma &:= \tau \exp(\int_0 \derG),\\
        \gamma &:= \tau \exp(\int_0\partial\res_{\mathbb{H}_2}h\derG)
    \end{align}
    where $h$ is the contracting homotopy from Theorem \ref{thm:acyclic} $i)$. Notice that we have
    \begin{align}
        \psi_M \circ \tilde\gamma\circ \rho = \psi_M \label{eqn:gamma-rho-preserves}
    \end{align}
    since $\tilde\gamma \circ \rho = \tau \exp( \int_0\partial\res_{\mathbb{H}_1}\dert^{(1)})$ and $\dert^{(1)}$ preserves $\psi_M$. This means that $\psi_M\circ \tilde\gamma$ is nothing but the family of states $\psi := \psi_M\circ \rho^{-1}$ on $M\times [0,2\pi]$ describing a flux-insertion at infinity which was used in Section \ref{sec:thouless-as-berry}. On the other hand, $\psi_M\circ \gamma$ inserts a flux at the origin in $\Z^2$ (see Figure \ref{fig:psigamma}). 
\begin{figure}[ht]
\centering
\includegraphics[width=0.5\textwidth]{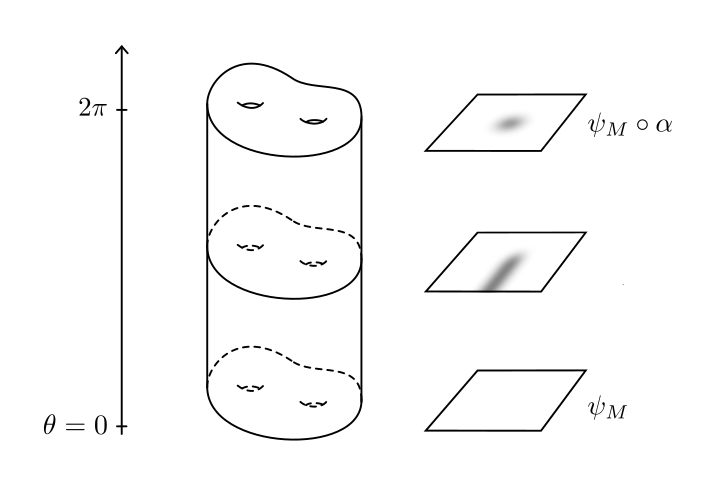}
\caption{$\psi_M\circ\gamma$ performs a flux insertion at the origin.}
\label{fig:psigamma}
\end{figure}
    We are now ready to begin the proof of Theorem \ref{thm:2dThoulessquantization}:
    \begin{proof}
    As before let $\{U_a\}_{a\in J}$ be an open cover of $M$ such all $U_a$ and all nonempty intersections $U_{ab}:=U_{a}\cap U_{b}$ are contractible. Let $\tilde\alpha$ and $\alpha$ be the families of LGAs on $M$ given by $\tilde\alpha:= \tilde\gamma \circ j_{2\pi}$ and $\alpha := \gamma\circ j_{2\pi}$, where $j_\theta: M\to M\times [0,2\pi]$ is the embedding $x\mapsto (x,\theta)$. 
    
    The family of states $\psi_M\circ\alpha$ differs from $\psi_M$ only near the origin. In fact, on each neighbourhood in $M$ the family $\psi_M\circ \alpha$ can be produced from $\psi_M$ by the action of an almost-local unitary, as we now show. Notice that we have $\psi_M \circ \alpha^{-1} = \psi_M\circ \tilde\alpha\circ \alpha^{-1}$, and that $\tilde\alpha \circ\alpha^{-1} = \tau\exp(\int_0^{2\pi} \alpha^{-1}(\partial \res_{\mathbb{H}_2^c}h\derG))$. Since $\derG_{vert}$ is confined on $\partial\mathbb{H}_1$ (Lemma \ref{lem:dergvert-confined}), by Lemma \ref{lem:alunitary} on each $U_a$ we can find a smooth map $V_a:U_a\to U(\mathscr{A}_{al})$ such that $\psi_M\circ \alpha^{-1} = \psi_M\circ \Ad_{V_a^{-1}}$ and $\overline{\tr}(V_a^{-1}dV_a)= 0$.
    
    On an overlap $U_{ab}$ we have the smooth unitary $V_a^{-1}V_b$ which preserves $\psi_M$, and we define
    \begin{align}
        h_{ab}:=\psi_M(V_a^{-1}V_b) \in C^\infty(U_{ab},U(1))
    \end{align}
    which satisfies the cocycle condition $h_{ab}h_{bc}= h_{ac}$. As in Section \ref{sec:berryquant} define $\beta_a:= \Ad_{V_a^{-1}}\circ \alpha$ and $\derB_a := \beta_a^{-1}d\beta_a$, both of which preserve $\psi_M$. Since $\beta_a\beta_b^{-1} = \Ad_{V_a^{-1}V_b}$, a straightforward calculation shows
    \begin{align}
	h_{ab}^{-1}dh_{ab} &= \psi_M(\derB_a - \derB_b). \label{eqn:hdh}
    \end{align}
    By Lemma \ref{lem:connection-interpolation}, $\psi_M$ is parallel with respect to a $\derC\in \Omega^1(M,\mfkDal)$ that smoothly interpolates between $\derG_M$ on $\mathbb{H}_2^c$ and $\derG_M^{\tilde\alpha}$ on $\mathbb{H}_2$. Defining
	\begin{align}
		a_a := -\psi_M(\derC - \derG_M^{\beta_a}),
	\end{align}
    we obtain, from (\ref{eqn:hdh}),
    \begin{align}
        h_{ab}^{-1}dh_{ab} &= \psi_M(\derB_a - \derB_b)  + \psi_M((\beta_a^{-1}(\derG_M) -\beta_b^{-1}(\derG_M))\nonumber\\
        &= a_a-a_b.
    \end{align}
    where the second line is because $\psi_M((\beta_a^{-1}(\derG_M) -\beta_b^{-1}(\derG_M)) = \psi_M((\derG_M - \Ad_{V_a^{-1}V_b}(\derG_M))=0$.
    Since $\psi_M$ is parallel with respect to $\frac12 (\derG_M^{\beta_a} + \derC)$ we have
    \begin{align}\label{eq:da}
        da_a &= -\psi_M(d(\derC - \derG_M^{\beta_a})+\frac12  \{\derG_M^{\beta_a} + \derC,\derC - \derG_M^{\beta_a}\})\nonumber\\
            &= -\psi_M(\derF_\derC - \beta_a^{-1}(\derF_{M})).
    \end{align}
    This collection of closed 2-forms is a \v{C}ech 0-cocycle in $\Omega^2(M,{i}\RR)$ and thus defines a closed 2-form on $M$. This can be seen in two different ways. First, one can compute the \v{C}ech coboundary using the expressions on the r.h.s.:
    \begin{equation}
\psi_M\left(\beta_a^{-1}\left(\derF_{M}\right)-\beta_b^{-1}\left(\derF_{M}\right)\right)=\psi(\derF_{M}-\beta_a(\beta_b^{-1}(\derF_{M})))=\psi(\derF_{M}-\Ad_{V_a^{-1}V_b}(\derF_{M}))=0,
    \end{equation}
    where we used the fact that the automorphisms $\beta_a^{-1},\beta_b^{-1}$ and $\Ad_{V_a^{-1} V_b}$ preserve $\psi_M$. Second, from the definition of $a_a$ we have $da_a-da_b = d(h_{ab}^{-1}dh_{ab})=0$. Thus $da_a$ is a restriction of a globally defined closed 2-form $\omega$. In addition, the second argument shows that the cohomology class of $\omega/2\pi{i}$ is integral (it is the de Rham representative of the 1st Chern class of the line bundle defined by the \v{C}ech 2-cocycle $h_{ab}$). 

    Finally, let us show that $-\omega$ is cohomologous to $2\pi\psi_M(\angles{\dert^{(3)},[*]})$. The strategy will be to define a form $b\in \Omega^2(M\times[0,1],\C)$ with $j_0^*b - j_{2\pi}^*b = -\omega$ and $db = \psi_M(\angles{\dert^{(3)},[*]})\wedge d\theta$. Then the result will follow from the following formula:
    \begin{align}
        -\omega &= j_{2\pi}^*b - j_0^*b \nonumber \\
        &= d\int_0^{2\pi} b - \int_0^{2\pi} db \nonumber \\
        &= d\int_0^{2\pi} b + 2\pi\psi_M(\angles{\dert^{(3)},[*]}). \label{eqn:ThoulessFinal}
    \end{align}
    Begin by defining
    \begin{align}
        \hat\derC := \derG - \partial\res_{\mathbb{H}_2}(h^\psi (\derG - \derG_M^{\tilde\gamma}) )\in\Omega^1(M\times [0,2\pi],\mfkDal),\\
        \derE :=  \partial \res_{\mathbb{H}_2^c}\derg^{(1)} + \partial \res_{\mathbb{H}_2} \tilde\gamma^{-1}(\derg_M^{(1)})\in\Omega^2(M\times [0,2\pi],\mfkDal),
    \end{align}
    where $\derg^{(1)}, \derg^{(1)}_M$ are as in Section \ref{sec:thouless-as-berry}. These definitions have been chosen so that $\hat\derC$ and $\derE$ satisfy the following properties:
    \begin{enumerate}[i)]
        \item $\psi$ is parallel with respect to $\hat\derC$
        \item $\hat\derC$ smoothly interpolates between $\derG$ on $\mathbb{H}^c_2$ and $\derG_M^{\tilde\gamma}$ on $\mathbb{H}_2$
        \item $\derF_{\hat\derC}$ smoothly interpolates between $\derF$ on $\mathbb{H}^c_2$ and $\tilde\gamma^{-1}(\derF_M)$ on $\mathbb{H}_2$
        \item $j_0^*\hat\derC = \derG_M$ and $j_{2\pi}^*\hat\derC=\derC$
        \item $\psi_M(j_{2\pi}^*\derE-\beta_a^{-1}(j_0^*\derE)) = 0$ for any $a\in J$
        \item $\derF_{\hat\derC}-\derE$ is smoothly confined at the origin in $\Z^2$.
    \end{enumerate}
    The first three follow from Lemma \ref{lem:connection-interpolation}, and the fourth is easy to verify. Property v) follows from the identity
    \begin{align}
        \psi_M(j_{2\pi}^*\derE-\beta_a^{-1}(j_0^*\derE)) = \psi_M(\partial\res_{\mathbb{H}_2}(\beta_a^{-1}(\derg^{(1)}_M) - \tilde\alpha^{-1}(\derg^{(1)}_M))).
    \end{align}
    The right-hand side of the above expression is well-defined since $\beta_a^{-1}(\derg^{(1)}_M) - \tilde\alpha^{-1}(\derg^{(1)}_M)$ is smoothly confined on $\mathbb{H}_1^c\cap \mathbb{H}_2$, and it is zero because both $\beta_a$ and $\tilde\alpha_a$ preserve $\psi_M$. Finally, property vi) is proved in the following:
    \begin{lemma}
         $\derF_{\hat\derC}-\derE$ is smoothly confined at the origin in $\Z^2$.
    \end{lemma}
    \begin{proof}
         Since both $\derF_{\hat\derC}$ and $\derE$ interpolate between $\derF$ on $\mathbb{H}^c_2$ and $\tilde\gamma^{-1}(\derF_M)$ on $\mathbb{H}^c$, it follows that $\derF_{\hat\derC}-\derE$ is smoothly confined on $\partial\mathbb{H}_2$. Next, using the fact that both $\derG - \derG_M$ and $\derG_M^{\tilde\gamma}- \derG_M$ are is smoothly confined on the vertical line $\partial \mathbb{H}_1$, one can show that $\derF_{\hat\derC}-\derF_M$ is too. Similarly, since both $\derg^{(1)}-\derg^{(1)}_M$ and $\tilde\gamma^{-1}(\derg^{(1)}_M)-\derg^{(1)}_M$ and are smoothly confined on $\partial\mathbb{H}_1$, one can show that $\derE-\derF_M$ is too. Thus $\derF_{\hat\derC}-\derE = (\derF_{\hat\derC}-\derF_M) -(\derE-\derF_M)$ smoothly confined on $\partial\mathbb{H}_1$. Thus it is smoothly confined on the origin.
    \end{proof}
    Property vi) allows us to define
    \begin{align}
        b := \psi(\derF_{\hat\derC}-\derE)\in\Omega^2(M\times [0,2\pi]).
    \end{align}
    \begin{lemma}
    With $b$ defined as above, we have
    \begin{align}
        j_0^*b - j_{2\pi}^*b = -\omega \label{eqn:sec5-jb}
    \end{align}
    and
    \begin{align}
        db = \psi(\angles{\derg^{(2)} -\rho(\derg_M^{(2)}), [\partial\mathbb{H}_2]}). \label{eqn:sec5-db}
    \end{align}
    \end{lemma}
    \begin{proof}
    First, from property v) above we have
    \begin{align}
        -\omega|_{U_a} &= \psi_M(\derF_{\derC} - j_{2\pi}^*\derE - \beta_a^{-1}(\derF_{M} - j_0^*\derE))\nonumber\\
        &= j_{2\pi}^*b|_{U_a} - j_0^*b|_{U_a},
    \end{align}
    and so (\ref{eqn:sec5-jb}) is established. Next, we have
    \begin{align}
        db &= \psi\left(D_{\hat{\derC}}\derF_{\hat\derC} + (D - D_{\hat\derC})\derF_{\hat\derC}  -D\derE\right)\nonumber\\
        &= \psi\left(\{\derG-\hat\derC,\derF_{\hat\derC}\}-D\derE\right)
    \end{align}
    where the second line is due to the Bianchi identity $D_{\hat{\derC}}\derF_{\hat\derC}= 0$. Inserting the definitions of $\hat\derC$ and $\derE$ into the above, and adding and subtracting the term $\tilde{\gamma}^{-1}\left(D_M\partial\res_{\mathbb{H}_2}\derg^{(1)}_M\right)$, we get
    \begin{align}
        db &= \psi\left(\{\derG-\hat\derC,\derF_{\hat\derC}\} + D\partial\res_{\mathbb{H}_2}\derg^{(1)} - D\partial\res_{\mathbb{H}_2}\tilde{\gamma}^{-1}\left(\derg_M^{(1)}\right)\right)\nonumber\\
        &= \psi\left(D\partial\res_{\mathbb{H}_2}\derg^{(1)}-\tilde{\gamma}^{-1}\left(D_M\partial\res_{\mathbb{H}_2}\derg^{(1)}_M\right)\right)\nonumber\\
        &\hspace{5mm} + \psi\left(\{\derG-\hat\derC,\derF_{\hat\derC}\} + \tilde{\gamma}^{-1}\left(D_M\partial\res_{\mathbb{H}_2}\derg^{(1)}_M\right) - D\partial\res_{\mathbb{H}_2}\tilde{\gamma}^{-1}\left(\derg_M^{(1)}\right)\right). \label{eqn:sec5-db-long}
    \end{align}
    The first term in (\ref{eqn:sec5-db-long}) is
    \begin{align}
        -\psi\left(\partial\res_{\mathbb{H}_2}\partial(\derg^{(2)} -\tilde{\gamma}^{-1}(\derg_M^{(2)}))\right) &= \psi(\angles{\derg^{(2)} -\tilde{\gamma}^{-1}(\derg_M^{(2)}), [\partial\mathbb{H}_2]}),
    \end{align}
    while the second one is
    \begin{align*}
        &\psi\left(\{\derG-\hat\derC,\derF_{\hat\derC}\} + (\tilde{\gamma}^{-1}\circ D_M\circ \tilde\gamma - D)\partial\res_{\mathbb{H}_2}\tilde{\gamma}^{-1}\left(\derg^{(1)}_M\right)\right)\\
        =&\psi\left(\{\derG-\hat\derC,\derF_{\hat\derC}\} + \left\{\derG_M^{\tilde\gamma} - \derG ,\partial\res_{\mathbb{H}_2}\tilde{\gamma}^{-1}\left(\derg^{(1)}_M\right)\right\}\right).
    \end{align*}
    Splitting $\derG_M^{\tilde\gamma} - \derG$ in the above into $\partial\res_{\mathbb{H}_2}h^\psi(\derG_M^{\tilde\gamma} - \derG) + \partial\res_{\mathbb{H}_2^c}h^\psi(\derG_M^{\tilde\gamma} - \derG)$ gives
    \begin{align*}
    &\psi\left(\left\{\partial\res_{\mathbb{H}_2}(h^\psi (\derG - \derG_M^{\tilde\gamma}) ),\derF_{\hat\derC}
    -\partial\res_{\mathbb{H}_2}\tilde{\gamma}^{-1}\left(\derg^{(1)}_M\right)\right\}\right)\\
    &\hspace{10mm} - \psi\left(\left\{ \partial\res_{\mathbb{H}_2^c}(h^\psi (\derG - \derG_M^{\tilde\gamma}) ),\partial\res_{\mathbb{H}_2}\tilde{\gamma}^{-1}\left(\derg^{(1)}_M\right)\right\}\right). 
    \end{align*}
    Both terms above are of the form $\psi(\{\derA,\derB\})$ where $\derA$ and $\derB$ preserve $\psi$ and are smoothly confined on regions with bounded stable intersection. Thus it is well-defined and zero, and so we have
    \begin{align}
        db &= \psi(\angles{\derg^{(2)} -\tilde{\gamma}^{-1}(\derg_M^{(2)}), [\partial\mathbb{H}_2]}).
    \end{align}
    Finally, (\ref{eqn:sec5-db}) follows from the following claim:
    \begin{align}
        \psi(\angles{\rho(\derg^{(2)}_M) - \tilde\gamma^{-1}(\derg^{(2)}_M),[\partial\mathbb{H}_1]}) = \psi_M(\angles{\derg^{(2)}_M - \rho^{-1}\circ\tilde\gamma^{-1}(\derg^{(2)}_M,[\partial\mathbb{H}_1]})) = 0. \label{eqn:lastclaim}
    \end{align}
    (the terms in the above equation are well-defined because $\rho\circ \tilde\gamma^{-1}(\derg^{(2)}_M) - \derg^{(2)}_M$ is smoothly confined on $\partial \mathbb{H}_1$ by $U(1)$-invariance of $\derg^{(2)}_M$). The first equality in (\ref{eqn:lastclaim}) is by definition of $\psi$, and the second follows from (\ref{eqn:gamma-rho-preserves}).
    \end{proof}
    We are now ready to conclude the proof of Theorem \ref{thm:2dThoulessquantization}. Looking at (\ref{eqn:g2def}) it is clear that $\derg^{(2)} - \rho(\derg^{(2)}) = \derg^{(2)}_{vert}$, and so by Theorem \ref{thm:thouless-as-berry} we have $db = \psi_M(\angles{\dert^{(3)}, [*]})\wedge d\theta$.
    By (\ref{eqn:ThoulessFinal}) this concludes the proof that the 2D Thouless form is quantized when $\psi$ is a smooth family of $U(1)$-invariant SRE states. To extend this to smooth families of $U(1)$-invariant invertible states is a matter of applying the same argument as the one appearing at the end of Section \ref{sec:berryquant}.
    \end{proof}

\begin{remark}
One can construct an example of a family of 2d states with a non-trivial 2d Thouless pump invariant in a way similar to the example from Section \ref{sec:berryquant}. The family is parameterized by $S^2$. At the north pole, it is a product state with vanishing ground-state $U(1)$ charges on each site. At the south pole, it is a product state with the ground-state  charge $(-1)^{x+y}$ on a site $(x,y) \in \ZZ^2$. There are four different pairings of neighbouring sites that correspond to four different directions $0, \pi/2, \pi$ and $3\pi/2$: 1) $(2k,l)$ with $(2k+1,l)$; 2) $(2k-1,l)$ with $(2k,l)$; 3) $(k,2l)$ with $(k,2l+1)$; 4) $(k,2l-1)$ with $(k,2l)$. There are also four different ways to form quadruples of neighbouring sites: $\{(2k,2l), (2k+1,2l), (2k,2l+1), (2k+1,2l+1)\}$ and its shifts. The meridians of $S^2$ at $0, \pi/2, \pi$ and $3\pi/2$ correspond to the families of states between the poles such that at each point the state is a product of pure states on pairs of sites for the corresponding pairing. Four different quarters of the sphere between the meridians correspond to the families of states such that at each point the state is a product of pure states on quadruples of sites. The family can be made smooth by choosing partitions of unity and has a unit 2d Thouless pump invariant. We omit the precise formulas.
\end{remark}
    
\pagebreak
\appendix
\appendixpage
\addappheadtotoc
\section{Asymptotically equal states}
It is a standard result in the theory of spin systems that two pure states on $\mathscr{A}$ unitarily equivalent if they are ``equal at infinity'' (see for instance Corollary 2.6.11 in \cite{BR1}). Below we prove two versions of this theorem that are adapted to our needs. Namely, first we show that if one of the states is SRE and the states rapidly approach each other at infinity, then the two states are related by the action of an almost-local unitary. Then we show that for certain smooth families of SRE states this almost-local unitary can be chosen to be a smooth function of parameters on any contractible neighbourhood in the parameter space. The first of these statements was proved in \cite{KSY}, but we include its proof here for completeness. 
\begin{lemma}\label{lem:unitary}
    Let $|\chi^1\rangle$ and $|\chi^2\rangle$ be two vectors in a Hilbert space $\mathcal{H}$ such that $\langle \chi^i|\chi^i\rangle=1$, $i=1,2$, and $\langle\chi^1|\chi^2\rangle >0$. Then there exists a unitary $U \in U(\mathcal{H})$ such that $U|\chi^1\rangle = |\chi^2\rangle$ and $\|U- \boldsymbol{1}\| = \||\chi^1\rangle - |\chi^2\rangle\|$.
\end{lemma}
We omit the proof of this elementary lemma.

\begin{prop}\label{prop:alunitary}
    Suppose $\psi$ is an SRE state and $\phi$ is another pure state such that there exists a superpolynomially decreasing function $f$ for which
    \begin{align}
    |\psi(A)-\phi(A)|\le f(R)\|A\| \label{eqn:prop-localunitary}
    \end{align}
    holds for any $A$ localized outside of $B_R(0)$. Then there is a unitary $V\in \mathscr{A}_{al}$ such that $\phi = \psi \circ \Ad_{V}$.
\end{prop}
\begin{proof}
    Since $\psi$ is SRE, we have $\psi = \psi_{fact}\circ \alpha^{-1}$ for a factorized pure state $\psi_{fact}$ and some LGA $\alpha$. Then $|\psi_{fact}(A)-\phi\circ\alpha|$ decays superpolynomially, as in (\ref{eqn:prop-localunitary}), and if we find a unitary $V\in \mathscr{A}_{al}$ with $\phi\circ\alpha = \psi_{fact}\circ \Ad_{V^{-1}}$ then we would have $\phi = \psi_{\alpha(V^{-1})}$. So we may assume without loss of generality that $\psi$ is a factorized state.
    
    Let $(\mathcal{H},\pi,|0\rangle)$ be the GNS triple of $\psi$. From (\ref{eqn:prop-localunitary}) and Corollary 2.6.11 in \cite{BR1} it follows that $\phi$ is given by a vector state $|\phi\rangle\langle\phi|$ in $\mathcal{H}$. For each positive integer $R$ let $B_R$ be the ball of radius $R$ around zero, $|0\rangle_{B_R}$ and $|0\rangle_{B_R^c}$ the (pure) restrictions of $|0\rangle$ to $B_R$ and $B_R^c$, respectively. Since $\psi$ is factorized we have $\mathcal{H} = \mathcal{H}_{B_R}\otimes \mathcal{H}_{B_R^c}$, where $\mathcal{H}_{B_R}$ and $\mathcal{H}_{B_R^c}$ are the GNS Hilbert spaces of $\psi|_{B_R}$ and $\psi|_{B_R^c}$.

    Pick an $R_0>0$ with $f(R_0)<1$ and let $R\ge R_0$. Notice that the purifications of $|0\rangle_{B_R^c}$ in $\mathcal{H}$ are precisely the unit-norm vectors in $\mathcal{H}_{B_R}\otimes |0\rangle_{B_R^c}$, and that $|\phi\rangle$ is a purification of $\phi|_{B_R^c}$ in $\mathcal{H}$. By Uhlmann's theorem \cite{Uhlmann} we have
    \begin{align}
        \max\left\{|\angles{\psi|\chi}| \ \bigg| \ |\chi\rangle \in \mathcal{H}_{B_R}\otimes |0\rangle_{B_R^c} \text{ and } \angles{\chi|\chi} = 1\right\} &= F(\phi|_{B_R^c}, \psi|_{B_R^c}). \label{eqn:chimax}
    \end{align}
    Let $|\chi^R\rangle$ be a maximizer with $\angles{\phi|\chi_x^R}\ge 0$. It satisfies
    \begin{align}
        \angles{\phi|\chi^R} &= F(\phi|_{B_R^c}, \psi|_{B_R^c})\nonumber\\
        &\ge 1- f(R)/2, \label{eqn:overlap}
    \end{align}
    where the second line is due to the Fuchs-Van de Graaf inequality. Since $f(R)<1$, $|\phi\rangle$ and $|\chi^n\rangle$ aren't orthogonal, and since $|\chi^R\rangle$ maximizes (\ref{eqn:chimax}) it follows that $|\chi^R\rangle$ is the normalized projection of $|\phi\rangle$ onto $\mathcal{H}_{B_R}\otimes |0\rangle_{B_R^c}$.

    By Lemma \ref{lem:unitary} there is a unitary $U^{R_0}$ localized on $B_{R_0}$ such that $U^{R_0}|0\rangle_{B_{R_0}} = |\chi^{R_0}\rangle$. Next, by (\ref{eqn:overlap}) we have
    \begin{align}
        \||\chi^{R-1}\rangle - |\chi^R\rangle\| &\le \||\chi^{R-1}\rangle - |\phi\rangle \| + \||\chi^R\rangle - |\phi\rangle \|\nonumber\\
        &\le 2\sqrt{f(R-1)}.
    \end{align}
    so for any $R\ge R_0$, Lemma \ref{lem:unitary} guarantees unitary $U^R$ localized on $B_R$ with $U^R|\chi^{R-1}\rangle = |\chi^R\rangle$, and
    \begin{align}
        \|U^R - \boldsymbol{1}\| = 2\sqrt{f(R-1)}.
    \end{align}
    Then 
    \begin{align}
        V:=\lim_{R\to \infty}U^{R}\hdots U^{R_0} = \sum_{R=R_0+1}^\infty(U^{R-1}-1)U^{R-1}\hdots U^{R_0}
    \end{align}
    is unitary, satisfies $V|0\rangle = |\phi\rangle$, and for any $S\ge R_0$ we have $\|V- \sum_{R=R_0+1}^{S}(U^{R-1}-1)U^{R-1}\hdots U^{R_0}\|\le \sum_{r\ge S} 2\sqrt{f(r)}$, so $V\in \mathscr{A}_{al}$.
\end{proof}
    For the following Lemma, let $\Gamma_1,\Gamma_2 \subset \Z^d$ be two regions such that $\Gamma_1 \subset \mathbb{H}_d$, $\Gamma_2\subset \mathbb{H}_d^c$, and $\Gamma_1\cup \Gamma_2$ has bounded stable intersection with $\partial\mathbb{H}_d$.
\begin{lemma}\label{lem:alunitary}
    Let $K$ be a contractible open subset of $\R^n$ for some $n\ge0$, and let $\gamma_1$ and $\gamma_2$ be families of automorphisms on $K$ such that each $\gamma_i$ is obtained as the path-ordered integral of a derivation-valued 1-form on $K\times [0,1]$ that is smoothly confined on $\Gamma_i$.
    Suppose that $(\psi,\derG)$ a smooth family of SRE states on $K$ such that $\psi \circ \gamma_1 = \psi \circ \gamma_2$. Then there is a smooth family of unitary observables $V \in C^{\infty} (K,\mathscr{A}_{al})$ with $\psi \circ \gamma_i = \psi \circ \Ad_{V}$ and $\overline{\tr}(V^{-1}dV) = 0$.
\end{lemma}
\begin{remark}
    This lemma is used twice in the text: once in Section \ref{sec:berryquant} with $\Gamma_1=L\subset \Z^1$ and $\Gamma_2 = R\subset\Z^1$, and once in Section \ref{sec:2dThoulessQuant} with $\Gamma_1 = (\partial\mathbb{H}_1) \cap \mathbb{H}_2\subset \Z^2$ and $\Gamma_2 = (\partial\mathbb{H}_1) \cap \mathbb{H}_2^c\subset \Z^2$.
\end{remark}
\begin{proof}
    Choose an arbitrary basepoint $x_* \in K$ and let $F:K\times [0,1]\to K$ be a smoooth nullhomotopy to the point $x_*$. Then $\beta:= \tau\exp(\int_0^1F^*\derG)$ satisfies $\psi_x = \psi_{x_*}\circ \beta_x$ for any $x\in K$. Define $\tilde\gamma_i := \beta\circ\gamma_i\circ\beta^{-1}$. Since $\gamma_i$ is of the form $\tau\exp(\int_0^1\derG)$ for a $\derG\in \Omega^1(K\times[0,1],\mfkDal)$ that is smoothly confined on $\Gamma_i$, the LGA $\beta\circ\gamma_i\circ\beta^{-1} = \tau\exp(\int_0^1\beta(\derG))$ is of that form too. By Proposition \ref{prop:localization-LGAs}, $\tilde\gamma_i^{-1}d\tilde\gamma_i$ is smoothly confined on $\Gamma_i$.

    Let $\chi:= \psi \circ \gamma_i\circ \beta^{-1}$. Then $\chi$ is parallel with respect to $\tilde{\gamma}_i^{-1}d\tilde\gamma_i$ for $i=1,2$. By Lemma \ref{lem:connection-interpolation}, $\chi$ is parallel with respect to a $\derH \in \Omega^1(K,\mfkDal)$ that is confined on $\Gamma_1\cup \Gamma_2$ and smoothly interpolates between $\tilde\gamma_2^{-1}d\tilde\gamma_2$ on $\mathbb{H}_2$ and $\tilde\gamma_1^{-1}d\tilde\gamma_1$ on $\mathbb{H}_2^c$. It follows that $\derH$ is smoothly confined on a bounded region, so by Proposition \ref{prop:smooth-lift} it has a lift to $\Omega^{1}(K,\mathscr{A}_{al})$, which by abuse of notation we will also call $\derH$. Then the path-ordered integral of $\derH$ along $F$ is an almost-local unitary $W:= \tau\exp(\int_0^1F^*\derH)$ and we have $\chi_x= \chi_{x_*}\circ\Ad_{W_x}$ for any $x\in K$. 
    
    By the definition of $\chi$ it is apparent that it asymptotically equals $\psi_{x_*}$ away from a stable intersection of $\Gamma_1$ and $\Gamma_2$, as in (\ref{eqn:prop-localunitary}). Thus by Proposition \ref{prop:alunitary} there is a unitary $V_*\in \mathscr{A}_{al}$ with $\chi_{x_*} = \psi_{x_*} \circ \Ad_{{V_*}}$. Letting $V_x := \beta_x^{-1}(V_*W_x)$ we have
    \begin{align}
        \psi_x\circ (\gamma_i)_x &= \chi_x \circ \beta_x\nonumber\\
        &= \chi_{x_*}\circ\Ad_{W_x}\circ \beta_x\nonumber\\
        &= \psi_{x_*} \circ \Ad_{V_*W_x}\circ \beta_x\nonumber\\
        &= \psi_x \circ \Ad_{V_x}.
    \end{align}
    Finally, since $\overline{\tr}(V^{-1}dV)$ is a closed 1-form and $K$ is contractible, there is a smooth function $g: K\to U(1)$ with $g^{-1}dg = \overline{\tr}(V^{-1}dV)$, and multiplying $V$ by $g^{-1}$ ensures that $\overline{\tr}(V^{-1}dV)=0$.
\end{proof}

\section{Confined chains}\label{sec:appendix-confined} 
The goals of this appendix is to prove Propositions \ref{prop:localization-properties}, \ref{prop:localization-LGAs}, and \ref{prop:smooth-lift}, as well as to introduce Lemma 
\begin{lemma}\label{lem:confined-1-chain}
    A 1-chain $\derf$ is confined on $X\subset \Z^d$ iff there is another 1-chain $\derg$ whose entries $\derg_j$ vanish outside $X$ such that $\partial \derg = \partial \derf$.
\end{lemma}
\begin{proof}
    For each $j\in X$ define $S_j:=\{k\in X^c : d(k,X) = d(k,j)\}$. Choose any total order on $\Z^d$ and define $\tilde{S}_j := S_j \backslash \bigcup_{k<j}S_k$. Then the $\tilde{S}_j$ are disjoint and $\bigsqcup_{j\in X}\tilde{S}_j = X^c$.
    Define the 1-chain $\derg$ by
    \begin{align}
        \derg_j := \left\{\begin{array}{cc}
            \derf_j + \sum_{k\in \tilde{S}_j}\derf_k & \text{ if }j\in X \\
            0 & \text{ otherwise.} 
        \end{array}\right. \label{eqn:locder1}
    \end{align}
    Then $\partial \derg = \partial \derf= \derF$ and it remains only to show that $\derg$ is a 1-chain. Since $\derf$ is confined on $X$ we have a superpolynomially decaying $h$ such that $\|\derf_k\| \le h(d(k,X))$. Thus for $k\in \tilde{S}_j$ we have $\|\derf_k\| \le h(d(k,j))$ and it follows that $\sum_{k\in \tilde{S}_j}\derf_k$ is absolutely convergent and the $\derg_j$'s have uniformly bounded norm. Finally let us show that $\sum_{k\in \tilde{S}_j}\derf_k$ is $f$-confined on $j$ for a function $f$ which is the same for all $j$. Since $\derf$ is a 1-chain there is a superpolynomially decaying $g_1$ such that every $\derf_k$ is $g_1$-confined at $k$. Let $r>0$ be even. For each $k\in \tilde{S}_j\cap B_j(r/2)$ pick $A_k\in \mathscr{A}_{B_k(r/2)}$ with $\|A_k-\derf_k\|\le g_1(r/2)$, and let $B:= \sum_{k\in \tilde{S}_j\cap B_j(r/2)}A_k\in \mathscr{A}_{B_j(r)}$. Then we have
    \begin{align}
       \|\sum_{k\in \tilde{S}_j}\derf_k - B\| &\le \sum_{k\in \tilde{S}_j\cap B_j(r/2)}\|\derf_k-A_k\| + \sum_{k\in \tilde{S}_j\cap B_j(r/2)^c}\|\derf_k\|\nonumber\\
       &\le r^dg_1(r) +\sum_{R\ge r/2} (2R)^d h(R) := g_2(r).
   \end{align}
   Thus $\derg_j$ is $(g_1+g_2)$-confined on $j$.
\end{proof}
Below we will often use the following family of seminorms on traceless almost-local observables:
\begin{lemma}[\cite{LocalNoether} Proposition D.1]
    Let $V\subset \mathscr{A}_{al}$ be the subspace of traceless observables. Then for any $j\in \Z^d$ the family of seminorms $\{\|\cdot\|_{j,\alpha}\}_{\alpha\in \N}$ on $V$ is equivalent to the family $\{\|\cdot\|^{br}_{j,\alpha}\}_{\alpha\in \N}$ given by
    \begin{align}
        \|A\|^{br}_{j,\alpha} := \sup_X\|A^X\|(1+\diam(\{j\}\cup X))^\alpha
    \end{align}
    where the supremum is taken over all bricks in $\Z^d$ and $\sum_ZA^Z$ is the brick decomposition of the inner derivation corresponding to $A\in \mathscr{A}_{al}$.
\end{lemma}
The following lemma shows that if an observable is $h$-confined on two faraway points $j,k\in \Z^d$ then its norm decays superpolynomially with $d(j,k)$.
\begin{lemma}\label{lem:decaying-norm}
    For any traceless $A\in \mathscr{A}_{al}$ and any positive integer $\alpha$ we have
    \begin{align}
         \|A\| \le 2^{2d+1}(\|A\|^{br}_{j,\alpha+2d+1} + \|A\|^{br}_{k,\alpha+2d+1})\left(\frac{d(j,k)}{2}\right)^{-\alpha}
    \end{align}
\end{lemma}
\begin{proof}
    Let $R:= d(j,k)$.
    \begin{align}
        \|A\| &\le \sum_{X\neq \emptyset}\|A^X\| \nonumber\\
        &\le \sum_{X\neq \emptyset}\min\bigg((1+\diam(X\cup\{j\}))^{-\alpha-2d-1}\|A\|^{br}_{j,\alpha+2d+1}, \ (1+\diam(X\cup\{k\}))^{-\alpha-2d-1}\|A\|^{br}_{k,\alpha+2d+1}\bigg)  \nonumber\\
        &\le (\|A\|^{br}_{j,\alpha+2d+1} + \|A\|^{br}_{k,\alpha+2d+1})\left( \sum_{\diam(X\cup\{j\}) \ge R/2}(1+\diam(X\cup\{j\}))^{-\alpha-2d-1}\right. \nonumber \\
        &\hspace{60mm}+ \left.
        \sum_{\diam(X\cup\{k\}) \ge R/2}(1+\diam(X\cup\{k\}))^{-\alpha-2d-1}\right).
    \end{align}
    Since there are at most $(2r)^{2d}$ bricks $X$ with $\diam(X\cup \{j\}) = r$ for any $r>0$, we get
    \begin{align}
        \sum_{\diam(X\cup\{j\}) \ge R/2}(1+\diam(X\cup\{j\}))^{-\alpha-2d-1} &\le \sum_{r\ge R/2}(2r)^{2d}(1+r)^{-\alpha-2d-1}\nonumber\\
        &\le 4^d \sum_{r\ge R/2}(1+r)^{-\alpha-1}\nonumber\\
        &\le 4^d \int_{r\ge R/2}^\infty r^{-\alpha-1}dr\nonumber\\
        &= 4^d \alpha^{-1}(R/2)^{-\alpha}\nonumber\\
        &\le 4^d(R/2)^{-\alpha}.
    \end{align}
\end{proof}

The following lemma shows that the definitions of confinement for derivations and chains are compatible.
\begin{lemma}\label{lem:locder}
    Let $\derF \in \mfkDal$ and $X\subset \Z^d$. The following are equivalent:
    \begin{enumerate}[i)]
        \item $\derF$ is confined on $X$.
        \item For every $\alpha \in \Z_{>0}$ there is a $C_\alpha$ such that
        \begin{align}
            \|\derF^Z\|(1+\diam(Z))^\alpha(1+d(Z,X))^\alpha \le C_\alpha
        \end{align} for every brick $Z$.
        \item The 1-chain $\derf_j := \sum_{X\ni j}\frac{1}{|X|}\derF^X$ is confined on $X$.
    \end{enumerate}
\end{lemma}
\begin{proof}
$i)\implies ii)$. Let $\derF = \sum_Z \derF^Z$ be the brick decomposition of $\derF$. Fix a brick $Z$. There is an operator $A$ supported in $Z$ with $\|A\|=1$ and $\|\derF_Z(A)\| = \|\derF_Z\|$. From this it follows that 
\begin{align}
    \|\derF_Z\| &= \|\derF_Z(A)\|\nonumber\\
    &= \|\overline{\tr}_{Z^c}\derF(A)\|\nonumber\\
    &\le \|\derF(A)\|.
\end{align}
Since $\derF$ is confined on $X$ we have
\begin{align}
    \|\derF_Z\| \le \|\derF(A)\| \le \sum_{j\in X}h_1(d(j,Z))
\end{align}
for some superpolynomially decaying function $h_1$. By Proposition C.1 in \cite{LocalNoether} we have 
\begin{align}
    \|\derF^Z\|\le 4^d \|\derF_Z\|\le 4^d\sum_{j\in X}h_1(d(j,Z))
\end{align}
Letting $h_2(R) := \sum_{r\ge R}2^dr^dh_1(r)$, we have
\begin{align}
    \sum_{j\in X}h_(d(j,Z)) &\le \sum_{\substack{j\in X\\ k \in Z}}h_1(d(j,k)) \nonumber \\
    &\le \sum_{k\in Z}h_2(d(k,X)) \nonumber \\
    &\le \diam(Z)^dh_2(d(Z,X)). \label{eqn:locder-1}
\end{align}
and so
\begin{align}
    \|\derF^Z\|\le 4^d \diam(Z)^dh_2(d(Z,X)).
\end{align}
If $\diam(Z)^dh_2(d(Z,X)) \le (1+ d(X,Z))^{-\alpha}(1+ \diam(Z))^{-\alpha}$ then we have 
\begin{align}
    \|\derF^Z\|(1+ d(X,Z))^\alpha(1+ \diam(Z))^{\alpha}\le 4^d.
\end{align}
Otherwise if $\diam(Z)^dh_2(d(Z,X)) > (1+ d(X,Z))^{-\alpha}(1+ \diam(Z))^{-\alpha}$ we have
\begin{align}
    \|\derF^Z\|(1+d(Z,X))^\alpha(1+ \diam(Z))^{\alpha} &\le \|\derF\|_{2\alpha + d}(1+\diam(Z))^{-\alpha - d}(1+d(Z,X))^\alpha\nonumber\\
    &\le \|\derF\|_{2\alpha + d}h_2(d(Z,X))(1+d(Z,X))^{2\alpha}\nonumber\\
    &\le \|\derF\|_{2\alpha + d}\sup_{r}(1+r)^{2\alpha}h_2(r).
\end{align}

$ii)\implies iii)$.
First, $\derf$ is a 1-chain because
\begin{align}
    \|\derf_j\|^{br}_{j,\alpha} = \sup_{Z\ni j}\frac{1}{|Z|}\|\derF^Z\|(1+\diam(Z))^\alpha \le \|\derF\|_{\alpha}.
\end{align}
To show $\derf$ is confined on $X$ we use the bound
\begin{align}
    \|\derf_j\| &\le \sum_{r>0}\sum_{\substack{Z \ni j \\ \diam(Z) = r}}r^{-1}\|\derF^Z\|.
\end{align}
Consider an arbitrary term in the above sum. When $r \le d(j,X)/2$ we have
\begin{align}
    \|\derF^Z\| \le C_{\alpha}(1+r)^{-\alpha}(1+ d(j,X)/2)^{-\alpha}
\end{align}
since $d(Z,X) + r \ge d(j,X)$. On the other hand when $r \ge d(j,X)/2$ we will use the bound
\begin{align}
    \|\derF^Z\|\le C_{\alpha}(1+r)^{-\alpha}.
\end{align}
Putting these together and using the fact that there are at most $d(r+1)^{d+1}$ bricks of diameter $r$ containing $j$, we have
\begin{align}
    \|\derf_j\| &\le C_\alpha(1+d(j,X)/2)^{-\alpha} \sum_{0<r\le  d(j,X)/2} d(1+r)^{-\alpha+d} +  C_{\alpha}\sum_{r> d(j,X)/2} d(1+r)^{-\alpha+d}.
\end{align}
When $\alpha \ge d+2$ we obtain
\begin{align}
     \|\derf_j\| &\le dC_{\alpha}(1+c)(1+d(j,X)/2)^{-\alpha+d+1}
\end{align}
where $c$ is a constant with $\sum_{r>R}(1+r)^{-\alpha+d}\le c R^{-\alpha+d+1}$. Thus $\sup_j\|\derf_j\|$ decays superpolynomially with $d(j,X)$.

$iii)\implies ii)$. Suppose $\derF = \partial \derf$ for some 1-chain $\derf$ confined on $X$. By Lemma \ref{lem:confined-1-chain} above, we may assume that the entries of $\derf$ vanish outside $X$. Let $Y\subset \Z^d$ be a finite region and let $A \in \mathscr{A}_{Y}$ be arbitrary. Since $\derf$ is a chain we have $\|[\derf_j,A]\| \le \|A\|h(d(j,A))$ for some superpolynomially decaying $h$. Summing these over all $j\in X$ we find that $\derF$ is $h$-confined on $X$.
\end{proof}
Recall that every $\derF\in \mfkDal$ has $\|\derF^X\|\le h(\diam(X))$ for some superpolynomially decreasing $h$. For such a derivation we have
\begin{align}
    \|\derF(A)\|\le C|\supp(A)|\|A\| \label{eqn:lga-der-bound}
\end{align}
for any strictly local observable $A$, where $C = d\sum_{R>0}R^dh(R)$. By integrating this bound one can see that it continues to hold when $\derF$ is replaced by an LGA $\alpha$.
\begin{lemma}\label{lem:alobserv}
     Let $F:\mathscr{A}_{al}\to\mathscr{A}_{al}$ be a linear map that satisfies the bound (\ref{eqn:lga-der-bound}). 
     Suppose $F$ is $h_2$-confined on a region $X$. If $A$ is an observable that is $h_1$-confined at a site $j$ then $\|F(A)\|$ is bounded by a superpolynomially decreasing function of $d(j,X)$ that depends only on $h_1$, $h_2$, and $\|A\|$. 
\end{lemma}
\begin{remark}
    If $F$ is trace-preserving then this implies that for any $p$-chain $\derf$, $F(\derf)$ is a $p$-chain that is confined on $X$ if $F$ is.
\end{remark}
\begin{proof}
    For any site $\ell$ and any $R>0$ we have
    \begin{align}
        \sum_{k\in B_\ell(r)^c}h_2(d(\ell,k)) \le \sum_{R > r}(2R)^dh_2(R) := h_3(r).
    \end{align}
    where evidently $h_3$ decays superpolynomially if $h_2$ does. It follows that for any bounded set $Y \subset \Z^d$ we have
    \begin{align}
        \sum_{k\in X}h_2(d(k,Y))\le |Y|h_3(d(X,Y)). \label{eqn:alobserv1}
    \end{align}
    Since $A$ is $h_1$-confined at $j$ we have $\|A - \overline{\tr}_{B_j(r)^c}A\|\le 2h_1(r)$. Defining $A_r:= \overline{\tr}_{B_j(r)^c}A- \overline{\tr}_{B_{j}(r-1)^c}A$ for $r\ge 1$ and $A_0:= \overline{\tr}_{\{j\}^c}A-\overline{\tr}(A)\boldsymbol{1}$
    we have  
    \begin{align}
        \|A_r\| &= \|\overline{\tr}_{B_j(r)^c}(A - \overline{\tr}_{B_{j}(r-1)^c}A)\|\nonumber\\
        &\le 2h_1(r-1)
    \end{align}
    for $r\ge 1$ while $\|A_0\|\le 2\|A\|$. Writing $R:= d(j,X)$, we have
    \begin{align}
        \|\derF(A)\| &\le \sum_{r\ge 0}\|\derF(A_r)\| \nonumber\\
        &\le \|\derF(A_0)\| + \sum_{1\le r < R}\|\derF(A_r)\| + \sum_{r\ge R}\|\derF(A_r)\| \nonumber\\
        &\le h_2(R)\|A\| + \sum_{1\le r < R}\sum_{k\in A}h_3(d(k,B_j(r)))\|A_r\| + \sum_{r\ge R}C (2r)^d h_1(r-1) \nonumber\\
        &\le h_2(R)\|A\| +\sum_{1\le r < R}(2r)^d h_1(r-1) h_3(R-r) + \sum_{r\ge R}C(2r)^d h_1(r-1), \label{eqn:alobserv2}
    \end{align}
    where we used (\ref{eqn:alobserv1}) in the last line.
\end{proof}

\begin{proof}[Proof of Proposition \ref{prop:localization-properties}]
    The proofs involving $\mathsf{a}\in C^p$ for some $p\ge 0$ will need to be split into cases according to whether $p=0$ or $p>0$, i.e. whether $\mathsf{a}$ is a derivation or a chain. 
    
    $i)$: When $\mathsf{a}\in C^1$ this follows from Lemma \ref{lem:confined-1-chain}. When $\mathsf{a}\in C^p$ for $p>1$ it follows from Lemma \ref{lem:decaying-norm}.
    
    $ii)$: Consider first the case $p>0$ and let $\derf$ be a $p$-chain that is $h_1$-confined on $X$ and $h_2$-confined on $X'$. Since $Y$ is a stable intersection of $X$ and $X'$ there is a $c>0$ such that $d(j,Y) \le c\max(d(j,X),d(j,X'))$. For any $j_1\hdots, j_k\in \Z^d$ and any $1\le \ell \le k$ we have
    \begin{align}
        \|f_{j_1,\hdots, j_k}\| &\le \min(h_1(d(j_\ell, X)), h_2(d(j_\ell,X')))\nonumber\\
        &\le g(\max(d(j_\ell, X),d(j_\ell,X')))\nonumber\\
        &\le g(c^{-1}d(j,Y))
    \end{align}
    where $g = \max(h_1,h_2)$. Next, if $p=0$ (ie. $\mathsf{a}$ is a derivation) the result follows from Lemma \ref{lem:alobserv} and the $p=1$ case.

    $iii)$: Suppose first that $\mathsf{a} \in C^p$ and $\mathsf{b}\in C^q$ for $p,q>0$. Then the bound $\|[\mathsf{a}_{j_1,\hdots,j_p}, \mathsf{b}_{j_{p+1},\hdots,j_{p+q}}]\| \le 2\|\mathsf{a}_{j_1,\hdots,j_p}\|\|\mathsf{b}_{j_{p+1},\hdots,j_{p+q}}\|$ shows that $\|\{\mathsf{a},\mathsf{b}\}_{j_1,\hdots,j_{p+q}}\|$ decays superpolynomially with both $d(j_i,X)$ and $d(j_i,X'))$ for any $1\le i \le p+q$, and so by part $i)$ it decays polynomially with $d(j_i,Y)$.
    
    Next, suppose $p=0$ and $q>0$. Then by Lemma \ref{lem:alobserv} and Proposition D.4 in \cite{LocalNoether}, $\|\{\mathsf{a},\mathsf{b}\}_{j_1\hdots, j_q}\|$ decays superpolynomially with both $d(j_i,X)$ and $d(j_i,X'))$ for any $1\le i \le q$, so again by part $i)$ we are done.

    Finally if $p=q=0$ then the result follows from the $p=0,q=1$ case together with Lemma \ref{lem:locder}.
\end{proof}
    Proposition \ref{prop:localization-LGAs} now follows easily from Proposition \ref{prop:localization-properties}.
\begin{proof}[Proof of Proposition \ref{prop:localization-LGAs}]
    The fact that $\alpha^{-1}d\alpha$ is smoothly confined on $X$ follows from the explicit expression (\ref{eqn:ada}).
    The second fact follows from the expression
    \begin{align}
        \alpha(\derF)-\derF = \int_0^1 \alpha_s(\iota_{\frac{\partial}{\partial s}}\derG(\derF))ds. \label{eqn:integral-alpha}
    \end{align}
    By Proposition \ref{prop:localization-properties} the integrand is confined on $X$, so $\alpha(\derF)-\derF$ is confined there too. Differentiating the equation (\ref{eqn:integral-alpha}) and using the fact that $\alpha^{-1}d\alpha$ and $\derG$ are confined on $X$, together with Proposition \ref{prop:localization-properties}, shows that the partial derivatives of $\alpha(\derF)-\derF$ are also confined on $X$. 
\end{proof}
We now move on to proving Proposition \ref{prop:smooth-lift}.
\begin{lemma}\label{lem:inner-bounded}
    Suppose $\derF\in \mfkDal$ is $h$-confined on a bounded set $X\subset \Z^d$. Then the sum $A:=\sum_{Z}\derF^Z$ is absolutely convergent in $\mathscr{A}_{al}$ and for any $j\in X$ we have $\|A\|^{br}_{j, \alpha} \le (1+\diam(X))^\alpha(4^d + C_{\alpha}\|\derF\|_{2\alpha+d})$ for some constants $C_{\alpha}$ that depend only on $\alpha$ and $h$.
\end{lemma}
\begin{proof}
    From the proof of Lemma \ref{lem:locder} we have constants $C'_{\alpha}$ depending only on $h$ and $\derF$ such that
    \begin{align}
        \|\derF^Z\| &\le C'_\alpha(1+\diam(Z))^{-\alpha}(1+d(X,Z))^{-\alpha}.
    \end{align}
    In fact, from the proof of Lemma \ref{lem:locder}, we see that these constants are of the form $C'_\alpha = 4^d + C_\alpha\|\derF\|_{2\alpha+d}$ for some constants $C_\alpha$ depending only on $h$. 
    
    For any $j\in X$ we have an inclusion $X\subset B_{j}(R)$, where we denoted $R=\diam(X)$ for brevity. Therefore  
    \begin{align}
        \|\derF^Z\| &\le C'_\alpha (1+\diam(Z))^{-\alpha} (1+d(Z,X))^{-\alpha}\nonumber\\
        &\le C'_\alpha(1+\diam(Z)+ d(Z,X))^{-\alpha\nonumber}\\
        &\le C'_\alpha (1+\diam(Z)+\max(0,d(j,Z)-R))^{-\alpha}\nonumber\\
        &\le (1+R)^\alpha C'_\alpha (1+\diam(Z)+d(j,Z))^{-\alpha}.
    \end{align}
    From this bound, and the fact that for sufficiently large $\alpha$ we have $\sum_{X}(1+\diam(X\cup\{j\})^{-\alpha}<\infty$ (the sum being over all bricks $X$) it follows that $\sum_Z\derF^Z$ is absolutely convergent in $\mathscr{A}_{al}$. Furthermore, 
    \begin{align}
        \|A^Z\|(1+\diam(\{j\}\cup Z))^\alpha=\|\derF^Z\|(1+\diam(\{j\}\cup Z))^\alpha\leq C'_\alpha (1+\diam(X))^\alpha .
    \end{align}
\end{proof}
We are now ready to prove Proposition \ref{prop:smooth-lift}.
\begin{proof}[Proof of Proposition \ref{prop:smooth-lift}]
Suppose $\derF = \ad_{A}$ for an antiselfadjoint $A\in \Omega^\bullet(M,\mathscr{A}_{al})$, and let $X = \{0\}\subset \Z^d$. Regarding $A$ as a 1-chain $\derf$ with $f_j$ equal to $A$ if $j=0$ and $0$ otherwise and applying Proposition \ref{prop:localization-properties} $i)$ we find that $\derF$ is smoothly confined on $X$.

Conversely, suppose $\derF$ is smoothly confined on $\derF$. Since this is a local statement we may assume without loss that $M=\R^n$. By Lemma \ref{lem:inner-bounded}, for each $x\in \R^n$ and each multi-index $\mu$ the sum $A_\mu(x):=\sum_Z\partial^\mu\derF^Z\in\mathscr{A}_{al}$ is well-defined and for any $\mu,\alpha,j$ the seminorm $\|A_\mu(x)\|_{j,\alpha}$ is a continuous function of $x$.

To show that $A\in C^\infty(U,\mathscr{A}_{al})$ it suffices to show that for any $\mu$ and any $0\le i \le n$ the equation $\partial_i A_\mu = A_{\mu+i}$ holds in $C^\infty(U,\mathscr{A}_{al})$.
For any brick $Z$ and any $\mu$ the expression $(\partial^\mu\derF)^Z = \partial^\mu(\derF^Z)$ is a smooth function $\R^n \to\mathscr{A}_{Z}$, and so for any $x\in \R^n$, $1\le i \le n$, and $h>0$ we have
\begin{align}
    \left\| \frac{\partial^\mu\derF^Z(x+he_i) - \partial^\mu\derF^Z(x)}{h} - \partial^{\mu+i}\derF^Z(x)\right\| &\le 
    \frac{1}{h}\int_{h_0=0}^hdh_0\|\partial^{\mu+i}\derF^Z(x+h_0e_i) - \partial^{\mu+i}\derF^Z(x)\|\nonumber\\
    &\le \sup_{0\le h_0\le h}\|\partial^{\mu+i}\derF^Z(x+h_0e_i) - \partial^{\mu+i}\derF^Z(x)\|\nonumber\\
    &\le h\sup_{0\le h_0\le h}\|\partial^{\mu+2i}\derF^Z(x+h_0e_i)\|.
\end{align}
and so for any $j\in \Z^d$ we have
\begin{align}
    \left\| \frac{A_\mu(x+he_i) - A_\mu(x)}{h} - A_{\mu+i}\right\|^{br}_{j,\alpha} \le h \sup_{0\le h_0\le h} \left\| A_{\mu+2i}(x+h_0e_i)\right\|^{br}_{j,\alpha},
\end{align}
which approaches $0$ as $h\to 0$.
\end{proof}
We conclude this appendix with a lemma which we will often use to create derivations that interpolate on the lattice between one derivation and another.
\begin{lemma}\label{lem:connection-interpolation}
    Suppose $\psi$ is a gapped family of states on $M$ that is parallel with respect to both $\derG_1 \in \Omega^1(M, \mfkDal)$ and $\derG_2 \in \Omega^1(M, \mfkDal)$. Then for any $X\subset \Z^d$ there exists $\derG_3\in \Omega^1(M,\mfkDal)$ such that $\psi$ is parallel with respect to $\derG_3$, and $\derG_3$ (resp. $\derF_{\derG_3}$) smoothly interpolates between $\derG_1$ (resp. $\derF_{\derG_1}$) on $X$ and $\derG_2$ (resp. $\derF_{\derG_2}$)on $X^c$. If in addition $\derG_1$ and $\derG_2$ are both smoothly confined on some $Y\subset \Z^d$, then $\derG_3$ and $\derF_{\derG_3}$ are smoothly confined there too.
\end{lemma}
\begin{proof}
    Since $\psi$ is parallel with respect to both $\derG_1$ and $\derG_2$, their difference lies in $\Omega^1(M,\mfkDal^\psi)$. Define
    \begin{align}
        \derG_3 &:= \derG_1 - \partial\res_{X^c}(h^\psi(\derG_1-\derG_2))\\
            &= \derG_2 - \partial\res_{X}(h^\psi(\derG_2-\derG_1)).
    \end{align}
    Since $\res_{X^c}(h^\psi(\derG_1-\derG_2))$ (resp. $\res_{X}(h^\psi(\derG_2-\derG_1))$) is smoothly confined on $X$ (resp. $X^c$), by Proposition \ref{prop:localization-properties} $i)$ and $ii)$, $\derG_3$ interpolates between $\derG_1$ on $X$ and $\derG_2$ on $X^c$.

    Next, we have
    \begin{align}
        \derF_{\derG_3} &= \derF_{\derG_1} - \partial\res_{X^c} D_\derG(h^\psi(\derG_1-\derG_2))\\
        &= \derF_{\derG_2} - \partial\res_{X} D_\derG (h^\psi(\derG_2-\derG_1)).
    \end{align}
    by the same reasoning as above $\derF_{\derG_3} -\derF_{\derG_1}$ (resp. $\derF_{\derG_3} -\derF_{\derG_2}$) is smoothly confined on $X^c$ (resp. $X$). 
\end{proof}

\bibliographystyle{plain} 
\bibliography{Bibliography.bib} 

\end{document}